\documentclass[aps,pra,a4paper,reprint,onecolumn,floatfix,superscriptaddress,notitlepage,usenames,dvipsnames,svgnames,table,nofootinbib,longbibliography]{revtex4-1}
\pdfoutput=1
\usepackage[utf8]{inputenc}
\usepackage[T1]{fontenc}
\usepackage{graphicx}
\usepackage{grffile}
\usepackage{wrapfig}
\usepackage{rotating}
\usepackage[normalem]{ulem}
\usepackage{amsmath}
\usepackage{textcomp}
\usepackage{amssymb}
\usepackage{capt-of}
\usepackage{verbatim}

\usepackage{parskip}
\usepackage{mathrsfs}
\usepackage[margin= 0.75in]{geometry}
\usepackage[braket, qm]{qcircuit}
\usepackage{physics}

\newcommand{\mc}[1]{\mathcal{#1}}

\newcommand{\End}{\textnormal{End}}
\newcommand{\vac}{|\vec{0}\rangle }
\newcommand{\bvac}{\langle \vec{0}|}
\newcommand{\cre}[1]{a^\dagger_{#1}}
\newcommand{\creta}[2]{a^\dagger_{#1}(\eta_{#2})}

\usepackage[english]{babel}
\usepackage{letltxmacro}
\LetLtxMacro{\ORIGselectlanguage}{\selectlanguage}
\makeatletter
\DeclareRobustCommand{\selectlanguage}[1]{%
  \@ifundefined{alias@\string#1}
    {\ORIGselectlanguage{#1}}
    {\begingroup\edef\x{\endgroup
      \noexpand\ORIGselectlanguage{\@nameuse{alias@#1}}}\x}%
}
\newcommand{\definelanguagealias}[2]{%
  \@namedef{alias@#1}{#2}%
}
\makeatother
\definelanguagealias{en}{english}

\usepackage[autostyle,english=american]{csquotes}
\MakeOuterQuote{"}

\usepackage{times}
\usepackage{bbm}
\usepackage{etoolbox}
\usepackage{tcolorbox} 
\usepackage{fleqn} 
\usepackage{tikz}
\usepackage{amsthm}
\usepackage{amssymb}
\usetikzlibrary{arrows.meta}
\usepackage{mathtools}

\newtheorem{theorem}{Theorem}
\newtheorem{proposition}{Proposition}

\newtheorem{lemma}{Lemma}
\newtheorem{corollary}{Corollary}

\newtheorem{remark}[theorem]{Remark}
\newtheorem{protocol}[theorem]{Protocol}

\usepackage{bm}
\usepackage{dsfont}
\usepackage[font=small,labelfont=bf,justification=justified,format=plain]{caption}
\usepackage{subcaption}

\usepackage{thmtools}
\usepackage{thm-restate}

\definecolor{blueviolet}{rgb}{0.54,0.17,0.89}
\definecolor{mycitecolor}{rgb}{0,0.7,0.1}
\usepackage[pdfusetitle]{hyperref}
\hypersetup{pdflang={English},colorlinks=true,linkcolor=ForestGreen,citecolor=BrickRed,urlcolor=RoyalBlue} 
\usepackage[capitalise]{cleveref}

\graphicspath{{figs/}}

\begin{document}
\title{Advancing Quantum Networking: Some Tools and Protocols for Ideal and Noisy Photonic Systems}

\author{Jason Saied}\thanks{jason.saied@nasa.gov}
\affiliation{Quantum Artificial Intelligence Laboratory, NASA Ames Research Center, Moffett Field, CA 94035, USA}

\author{Jeffrey Marshall}
\affiliation{Quantum Artificial Intelligence Laboratory, NASA Ames Research Center, Moffett Field, CA 94035, USA}
\affiliation{USRA Research Institute for Advanced Computer Science, Mountain View, CA 94043, USA}

\author{Namit Anand}
\affiliation{Quantum Artificial Intelligence Laboratory, NASA Ames Research Center, Moffett Field, CA 94035, USA}
\affiliation{KBR, Inc., 601 Jefferson St., Houston, TX 77002, USA}

\author{Shon Grabbe}
\affiliation{Quantum Artificial Intelligence Laboratory, NASA Ames Research Center, Moffett Field, CA 94035, USA}

\author{Eleanor G. Rieffel}
\affiliation{Quantum Artificial Intelligence Laboratory, NASA Ames Research Center, Moffett Field, CA 94035, USA}

\date{\today}
\begin{abstract}
Quantum networking at many scales will be critical to future quantum technologies and experiments on quantum systems. Photonic links enable quantum networking. They will connect co-located quantum processors to enable large-scale quantum computers, provide links between distant quantum computers to support distributed, delegated, and blind quantum computing, and will link distant nodes in space enabling new tests of fundamental physics. Here, we discuss recent work advancing photonic tools and protocols that support quantum networking. We provide analytical results and numerics for the effect of distinguishability errors on 
key photonic circuits; we considered a variety of error models and developed new metrics for benchmarking the quality of generated photonic states. We review a distillation protocol by one of the authors that mitigates distinguishability errors. We also review recent results by a subset of the authors on the efficient simulation of photonic circuits via approximation by coherent states. We study some interactions between the theory of universal sets, unitary t-designs, and photonics: while many of the results we state in this direction may be known to experts, we aim to bring them to the attention of the broader quantum information science community and to phrase them in ways that are more familiar to this community.  We prove, translating a result from representation theory, that there are no non-universal infinite closed $2$-designs in $U(\mc{V})$ when $\dim V \geq 2$. As a consequence, we observe that linear optical unitaries form a $1$-design but not a 2-design. Finally, we apply a result of Oszmaniec and Zimbor\'{a}s to prove that augmenting the linear optical unitaries with any nontrivial SNAP gate is sufficient to achieve universality. 
\end{abstract}


\maketitle
\tableofcontents

\section{Introduction}

Quantum networking at many scales will be critical to future quantum technologies and experiments on quantum systems. The properties of light mean that quantum network links will be photonic, whether they connect co-located quantum processors to enable large-scale quantum computers; distant quantum computers to support distributed, delegated, and blind quantum computing; or distant nodes in space enabling new tests of fundamental physics. Here, we discuss recent results, tools and protocols for advancing robust photonic systems that support quantum networking. 

These results and tools were developed in the course of examining use cases for quantum networks, and in evaluating the resources that would be required to implement quantum networks that could support these use cases. The use cases we considered included photonic architectures for quantum computing, distributed quantum computing including extending results on distributed algorithms for distributed data \cite{kerger2023mind}, design of Wigner friend based experiments for future tests of fundamentals of quantum theory \cite{wiseman_thoughtful_2022}, and delegated quantum computing protocols including blind quantum computation. Our results and tools include techniques for the simulation of photonic information processing systems, particularly simulations under realistic errors, techniques for mitigating distinguishability errors, the effect of distinguishability errors on 
fusion \cite{browne_resource-efficient_2005} and $n$-GHZ state generation protocols \cite{bartolucci2021creation}, 
and advances in theory for benchmarking photonic quantum systems. 
Some of our results may be known to experts, but have yet to appear in the quantum information sciences literature. 
Part of the aim of this paper is to bring such results to the attention of the broader quantum information science community and to phrase them in ways that are accessible to this community. 

The key contributions of this paper include:
\begin{itemize}
\item Overviews of Fock state methods, and a coherent state method developed by a subset of the current authors \cite{marshall_simulation_2023}, for simulation of photonic quantum systems particularly under noise. 
These are given in Secs.~\ref{sec:loss}, \ref{sec:distinguishability}, and \ref{sec:coherent}. 
\item In Sec.~\ref{sec:distillation}, an overview of a protocol by one of the current authors for distilling a set of more indistinguishable photons from a larger set of less distinguishable photons \cite{marshall_distillation_2022}, reducing distinguishability errors at the outset.
\item Evaluation of the effect of distinguishability errors on circuits for Type I and II fusion \cite{browne_resource-efficient_2005}, generalized $n$-GHZ state measurements, and the generation of $n$-GHZ entangled states \cite{bartolucci2021creation}. 
The analytical results are given in Secs.~\ref{sec:fusion_dist_statements} and \ref{sec:dist ghz}. 
For the case of Bell pairs (GHZ states with $n=2$), we have provided detailed numerics in Section~\ref{sec:numerics}; these numerics considered a variety of error models and developed new metrics for benchmarking the quality of generated photonic states. 
\item A proof, translating a result from representation theory \cite{katz2004larsen}, that there are no non-universal infinite closed $2$-designs in $U(d)$ for $d\geq 2$. As a corollary, we observe that the linear optical unitaries form a $1$-design but not a $2$-design. Finally, we apply a result of Oszmaniec and Zimbor\'{a}s \cite{oszmaniec2017universal} to prove that augmenting the linear optical unitaries with any nontrivial SNAP gate is sufficient to achieve universality. 
These results are in Sec.~\ref{sec:designs supersection}. 
\item In Sec.~\ref{sec:discussion}, a brief overview of two future directions related to this work. First we discuss the applicability to larger systems of our analysis techniques and results regarding distinguishability errors. We propose a low-order approximation of the effect of distinguishability errors, as elaborated upon in Remark~\ref{rem:stabilizer sims}. We then consider the problem of how many SNAP gates are required to promote linear optics to an approximate $2$-design and provide some numerical results for small system sizes. 
\end{itemize}

\section{Distinguishability and loss errors}\label{sec:loss dist}
Loss and distinguishability are two of the most dominant contributions to noise and errors in optical settings, both of which are known to reduce the classical complexity of computational tasks such as Boson sampling \cite{BS-loss-aaronson, BS-distinguishable-renema}. As such, it is important to be able simulate the effects of each in software. 

Our starting point is Fock basis simulation. 
Generally we will be interested in simulating an $m$ mode system with a fixed number $n$ of (for now) identical photons. 
As such the Hilbert space can be represented via the second quantization as $\mathcal{H}=\mathrm{span}\{|n_1, \dots, n_m\rangle \}_{\sum_i n_i = n}$, 
with dimension $d_{n,m}:=\binom{n+m-1}{n}$. 
The initial state is typically composed of $n$ single photons in $m\ge n$ modes, $|1\rangle^{\otimes n}|0\rangle^{\otimes (m-n)}$. 
Evolution occurs via a linear optical network, which maps creation operators\footnote{$a^\dag |n\rangle = \sqrt{n+1}|n+1\rangle$.} to a linear combination
\begin{align}
    U a_i^\dag U^\dag = \sum_{j=1}^m u_{ji} a_j^\dag,
\end{align}
where $u$ is an $m\times m$ unitary matrix, called the transfer matrix. It is known that any linear optical transformation can be implemented, to arbitrary accuracy, by a sequence of $O(m^2)$ beamsplitters\footnote{In fact,  any \textit{single} non-trivial beamsplitter (with not all entries real) can generate all of linear optics \cite{boulandGenerationUniversalLinear2014} (with $O(m^2)$ of these still required).} \cite{reckExperimentalRealizationAny1994}.

The full state output of such a network can be computed in time  $O(n d_{n,m})$, and space $O(d_{n,m})$ \cite{heurtelStrongSimulationLinear2023, marshall_simulation_2023}, by multiplying the evolved creation operators:
\begin{align}
    U|1\rangle^{\otimes n}|0\rangle^{\otimes(m-n)} = U \prod_{i=1}^n a_i^\dag U^\dag |\vec{0}\rangle = \prod_{i=1}^n \left(\sum_{j=1}^m u_{ji}a_j^\dag \right) |\vec{0}\rangle.
\end{align}
In the first step we write $|1\rangle = a^\dag |0\rangle$, and use that any linear optical unitary acting on the vacuum is the identity, i.e. $U^\dag |\vec{0}\rangle = |\vec{0}\rangle$. In the second step we insert resolutions of the identity $I = U^\dag U$ after each creation operator.

In this work we will assume measurements are all of the photon-number-resolving (PNR) type, which, in the absence of error, counts the number of photons in each measured mode. That is, it is equivalent to performing a measurement in the Fock basis.

In Secs. \ref{sec:loss} and \ref{sec:distinguishability}, we will discuss how to modify such a simulation to include losses and distinguishability errors.
This is followed by Sec.~\ref{sec:distillation}, in which we review work of Marshall \cite{marshall_distillation_2022} that gives a protocol for distilling less distinguishable photons. 
This protocol was discovered using the simulation techniques of Sec.~\ref{sec:distinguishability}. 
Then in Sec.~\ref{sec:distinguishability impact}, we consider the effect that distinguishability has on circuits for fusion, GHZ state generation, and related operations. 
The first several subsections give analytical results; Sec.~\ref{sec:numerics} gives numerics for Bell state generation in the presence of distinguishability, using the techniques of Sec.~\ref{sec:distinguishability}.

\subsection{Simulation of Loss}\label{sec:loss}
Loss has several physical mechanisms, such as absorption in optical fiber and components, imperfect photon sources, and in photo-detection \cite{photonic-quantum-information}. 
Loss is a non-unitary process, resulting in a mixed state over different lost photon number sectors \cite{oszmaniec_classical_2018}. We can write its action on a matrix element as
\begin{align}
    \Lambda_\eta(|n_1\rangle \langle n_2|) = \eta^{\frac{n_1+n_2}{2}}\sum_{k=0}^{\min\{n_1,n_2\}} \sqrt{\binom{n_1}{k}\binom{n_2}{k}}\left(\frac{1-\eta}{\eta}\right)^k |n_1-k\rangle \langle n_2 - k|.
    \label{eq:loss_channel}
\end{align}
Here $1-\eta$ is the single photon loss probability. 

It can readily be observed that the process Eq.~\eqref{eq:loss_channel} is equivalent to applying a beamsplitter between the mode of interest and a fictitious mode containing 0 photons, where the transmisivity (probability for photon to transition) is $1-\eta$ \cite{oszmaniecRandomBosonicStates2016}. Consider a beamsplitter that acts on modes $a^\dag, b^\dag$ as $a^\dag \rightarrow (\sin \theta a^\dag + \cos \theta b^\dag).$ Then we can see
\begin{align}
    |n,0\rangle = \frac{1}{\sqrt{n!}}(a^\dag)^n|0,0\rangle \rightarrow \frac{1}{\sqrt{n!}}(\sin \theta a^\dag + \cos \theta b^\dag)^n|0,0\rangle = \sum_{k=0}^n \sqrt{\binom{n}{k}} \sin^{n-k}(\theta) \cos^{k}(\theta)|n-k, k\rangle.
    \label{eq:loss-beamsplitter}
\end{align}
Upon tracing out the fictitious mode we get the same channel, with $\eta = \sin^{2}\theta$.

In practice this is simulated by sampling and keeping in memory only a single pure state. Consider for simplicity, an arbitrary $n$ photon pure state of $m$ modes, $|\psi_n\rangle$. We append a vacuum mode and apply a beamsplitter between it and the lossy mode of interest, which splits the state into different photon number sectors:
\begin{align}\label{eq:mixed loss}
    |\psi_n\rangle \rightarrow \sum_{k=0}^n c_k |\psi_{n-k}\rangle|k\rangle \rightarrow \sum_k |c_k|^2 |\psi_{n-k}\rangle \langle \psi_{n-k}|.
\end{align}
Here the $c_k$ are complex amplitudes ($\sum_k |c_k|^2=1$) that depend on $\theta$ of Eq.~\eqref{eq:loss-beamsplitter}, and $|\psi_{n-k}\rangle$ is notation to indicate it is \textit{some} $n-k$ photon state. After this step we can probabilistically sample $|\psi_{n-k}\rangle$ from the distribution $\{|c_k|^2\}_k$ (tracing out the fictitious mode). 

Another useful observation is that symmetric loss commutes through linear optics \cite{oszmaniec_classical_2018}. For example, $\Lambda_\eta^{(1)}\Lambda_\eta^{(2)}B_{12} = B_{12}\Lambda_\eta^{(1)}\Lambda_\eta^{(2)}$, where $B_{12}$ is a beamsplitter between modes 1,2, and $\Lambda_\eta^{(i)}$ is a loss channel on mode $i$. Coupling this with fact\footnote{We see this by applying Eq.~\eqref{eq:loss_channel} twice with different parameters: Consider the term that has $p$ photons subtracted after both channels. The coefficient is $(\eta_1\eta_2)^{(n_1+n_2)/2}\sum_{k=0}^p \sqrt{\binom{n_1}{k}\binom{n_2}{k}\binom{n_1-k}{p}\binom{n_2-k}{p}}\left(\frac{1-\eta_1}{\eta_1}\right)^k\left(\frac{1-\eta_2}{\eta_2}\right)^{p-k}\eta_2^{-k} = \dots = (\eta_1\eta_2)^{(n_1+n_2)/2} \sqrt{\binom{n_1}{p}\binom{n_2}{p}}\left(\frac{1-\eta_1\eta_2}{\eta_1\eta_2}\right)^p$, i.e. the same as a single loss channel with parameter $\eta=\eta_1\eta_2$.} $\Lambda_{\eta_1}\Lambda_{\eta_2}=\Lambda_{\eta_1\eta_2}$, certain circuits with symmetry can be simplified dramatically by commuting the losses. 
For an example, see Fig.~\ref{fig:loss-commute}, where all loss can be commuted to the beginning of the circuit. 
In such a case, following \eqref{eq:mixed loss}, we may simply sample input states with fewer photons rather than actually applying a loss channel. 
Further, if the probabilities $|c_k|^2$ of \eqref{eq:mixed loss} are known explicitly, we may often perform calculations for each possible input $\ketbra{\psi_{n-k}}$ separately, then use linearity to calculate the desired quantities. (See Sec.~\ref{sec:numerics} for this technique applied to the case of distinguishability.)

In a low-order error model, where error rates are low and therefore multiple nearby errors are assumed to be rare, 
photon loss is generally heralded. Thus in the calculations below we will generally ignore loss; however, for a more comprehensive treatment involving higher-order error terms, it will be essential to take loss into account. 

\begin{figure}
    \centering
    \includegraphics[width=0.75\columnwidth]{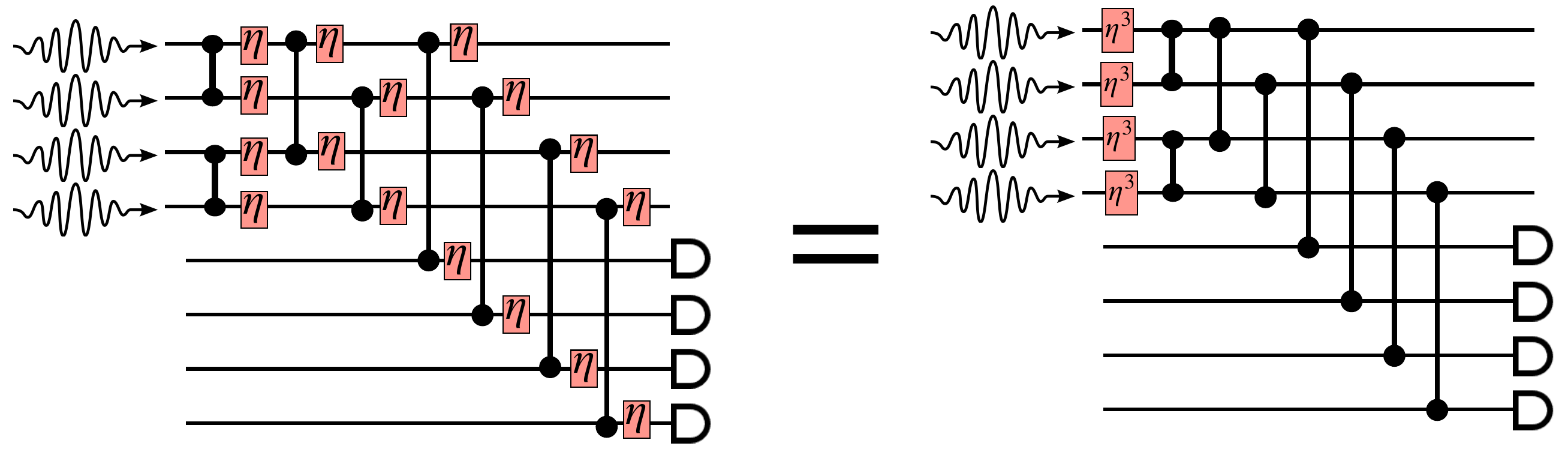}
    \caption{Example of circuit with symmetric losses that commute to the start of the circuit. Instead of applying 16 independent loss channels, it is equivalent to applying 4 loss channels with a modified parameter $\eta \rightarrow \eta^3$. In this diagram, the two qubit `gates' are beamsplitters, and the boxes with enclosed $\eta$ represent a loss channel with loss rate $1-\eta$.}
    \label{fig:loss-commute}
\end{figure}

\subsection{Simulation of Distinguishability}\label{sec:distinguishability}
Distinguishability typically arises from single photons generated by different sources, or by dispersion e.g., by travelling in fiber \cite{fanEffectDispersionIndistinguishability2021}. As the name suggests, (partial) photon distinguishability is characterized by two photons whose internal states are not identical, leading to non-ideal overlap $|\langle \psi|\phi\rangle|^2<1$. For example, two states with spectra sharply peaked at different wavelengths will be almost perfectly distinguishable. On the other hand, truly identical photons are perfectly indistinguishable, and no experiment could tell them apart.

The degree of distinguishability is typically measured by a Hong-Ou-Mandel (HOM) type experiment. Current `HOM dip' experiments shows the achievable visibility\footnote{The visibility between two pure states is the overlap squared, $V:=|\langle \phi|\psi\rangle|^2$.} between independent sources is in the regime of 90-99\%, depending on the source \cite{wangResearchHongOuMandelInterference2019, ollivierHongOuMandelInterferenceImperfect2021}. 

To begin a more detailed discussion, let's pick a simple model where two independent photon sources each output single photons in deterministic state $|\psi_0\rangle, |\psi_1\rangle$. Without loss of generality we can write $|\psi_1\rangle = c_0|\psi_0\rangle + c_1|\psi_0^\perp\rangle$, where $|c_0|^2+|c_1|^2=1$. Typically we are interested in the case where $\epsilon:=1-|c_0|^2$ is `small'. We can now represent this a la the second quantization as $|\psi_1\rangle=c_0|1,0\rangle + c_1|0,1\rangle = (c_0 a_0^\dag + c_1a_1^\dag)|0,0\rangle$. As an illustration of the effect this has, consider performing the HOM experiment with these states, where the 50:50 beamsplitter takes convention $a^\dag \rightarrow (a^\dag + b^\dag)/\sqrt{2}, b^\dag \rightarrow (a^\dag - b^\dag)/\sqrt{2}$. Notice that we now have four modes to consider; two spatial modes (e.g. the optical fibres) which we will label $a,b$, and the two internal modes, labeled $0,1$. The four relevant creation operators are therefore $a_0^\dag, a_1^\dag, b_0^\dag, b_1^\dag$. 

The two photon initial state evolves as
\begin{align}
    a_0^\dag (c_0 b_0^\dag + c_1 b_1^\dag) \rightarrow (a_0^\dag + b_0^\dag)(c_0 (a_0^\dag - b_0^\dag) + c_1 (a_1^\dag - b_1^\dag))/2 = \frac{c_0}{2} ((a_0^\dag)^2 - (b_0^\dag)^2) + \frac{c_1}{2}(a_0^\dag a_1^\dag - a_0^\dag b_1^\dag + b_0^\dag a_1^\dag - b_0^\dag b_1^\dag).
\end{align}
For $c_1=0$ we get the regular HOM result with photon bunching in either spatial mode. In the general case however, we see we now detect coincidence events (observing photons in both the $a$ and $b$ spatial modes at the same time), with probability $|c_1|^2/2$.

Under the reasonable assumption that no transformation of interest (linear optics and PNRD) will make two internal modes less distinguishable\footnote{For example, consider applying some (possibly non-unitary) transformation $\mathcal{E}$ on $|\psi_0^\perp\rangle$. The statement in the main text is equivalent to $\langle \psi_0| \mathcal{E}(|\psi_0^\perp\rangle \langle \psi_0^\perp|)|\psi_0\rangle=|\langle \psi_0|\psi_0^\perp\rangle|^2=0.$} (e.g. transforming $|\psi_0^\perp\rangle$ to $|\psi_0\rangle$), statistics generated by the above model are identical to those from a mixed state picture, where the state with the `error' is
\begin{align}\label{eq:dist error model}
    \rho(\epsilon) = (1-\epsilon) |\psi_0\rangle \langle \psi_0| + \epsilon |\psi_0^{\perp} \rangle \langle \psi_0^{\perp} |.
\end{align}
This is motivated by the fact that in linear optics, orthogonal internal states will never interfere with each other and so the relative magnitude of each `sector' remains unchanged. This can also be straightforwardly generalized to handling different numbers of `error mode', e.g. $\sum_i \epsilon_i |\psi_i\rangle \langle \psi_i|$, where $\langle \psi_i|\psi_j\rangle = \delta_{ij}$.  Two common models of multi-photon distinguishability are known as the \textit{Orthogonal Bad Bits} (OBB) model and the \textit{Random Source} model \cite{sparrowQuantumInterferenceUniversal2017}. 
We discuss various error models in more detail in Sec.~\ref{sec:numerics} below: namely, we consider the OBB model and two others that we call the Same Bad Bits (SBB) and Orthogonal Bad Pairs (OBP) models. 

\begin{figure}
    \centering
    \includegraphics[width=0.6\columnwidth]{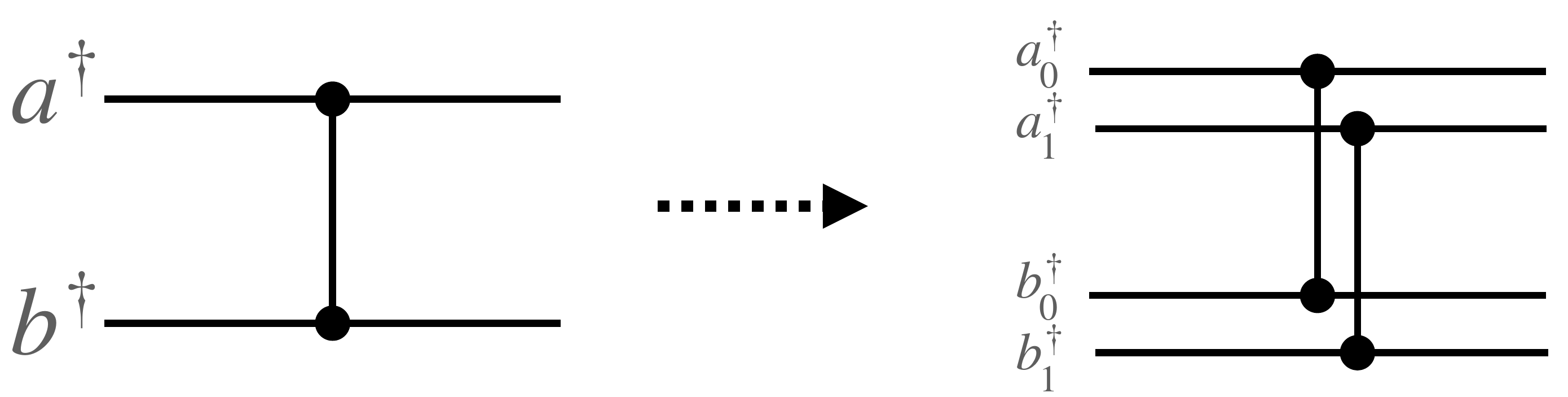}
    \caption{Simulating photons with two internal modes. By `copying' the circuit, additional internal (noise) modes can be simulated. As each circuit copy is independent, these simulations can be carried out separately. }
    \label{fig:circ_copy}
\end{figure}

Under these same assumptions, we can decompose the Hilbert space as a direct sum of non-interacting pieces, $\mathcal{H} = \oplus_{i=0}^k \mathcal{H}_i$, where there are $k$ different error types. In particular, $\mathcal{H}_i$ is the Hilbert space for photons of `type' $i$ (and we let $i=0$ index the `ideal' photon subspace). Operations in this framework are of the form $M = \oplus_i m_i$, which for state agnostic operations \cite{gonzalez-andradePolarizationWavelengthagnosticNanophotonic2019} reduces to $m_i=m,\forall i$. In this case, a simulation with $k$ errors can be visualized by copying the circuit of interest $k$ times, i.e., each physical rail can host $k+1$ internal modes and represented in the basis $\{|n_0, \dots, n_k\rangle\}$. A simple example with a single beamsplitter and one error mode is shown in Fig.~\ref{fig:circ_copy}.

To handle PNR measurements, we perform a POVM over \textit{all} modes, i.e., $\{|n_0, \dots, n_k\rangle \langle n_0, \dots n_k|\}$, for a single PNRD. However, we must in general do further post-processing, since the internal mode index sampled is typically not visible to the experimenter\footnote{To be clear, in principle, an internal mode index \textit{could} be observable (e.g. related to the frequency of the photon), however standard PNRDs or click/no-click detectors, as we assume in this work, do not resolve this information.} (e.g., in a single photon experiment, we would not learn which $i=0, \dots, k$ was detected). Indeed, we'll assume only the total number of photons $N=\sum_{i=0}^k n_i$ is accessible. Post-selection on a particular pattern can therefore have many associated `microscopic' configurations that must be traced out, in general leading to mixedness.

In terms of simulation, we can treat each `error subspace' entirely separately, and simulate independently the photons in each class. For example, if there is one error injected along some optical fibre, we will simulate separately the $n-1$ `ideal' photons, and then the 1 photon error state. At the measurement step we must perform post-processing to determine which configurations on each state can yield a desired measurement outcome, as discussed above.

So far we have mostly been considering the case where distinguishability between photons is introduced at the source(s). In LOQC however, it is typical that certain photons may propagate several `clock cycles', whereas others are introduced later. Via dispersion, this can introduce distinguishability, even if the different photon sources are perfect \cite{fanEffectDispersionIndistinguishability2021}. As such it is important to also consider introducing photon distinguishability during the execution of a simulation, for example transferring photons from one internal mode to another.

\subsection{Distinguishability distillation}\label{sec:distillation}

 It is highly desirable to have a heralded source of highly pure and indistinguishable photons. There are two main methods to mitigate distinguishabiilty effects; engineering the source, and unheralded spectral filtering. The first is technically challenging, often requiring physics insights, and the latter is often not acceptable for use in LOQC related tasks. Photon distillation \cite{sparrowQuantumInterferenceUniversal2017,marshall_distillation_2022} is a potential solution to this problem, in that it can be used to generate (in principle) heralded single photons with arbitrary indistinguishability.
 Moreover, it can be implemented with standard linear optics and PNRD. There is of course a cost to this, and that is in the number of photons that must be utilized in order to distill a single purer one.
 
 The most efficient scheme currently known is that of Ref.~\citenum{marshall_distillation_2022} which requires $O((\epsilon/\epsilon')^2)$ photons to distill a state from distinguishability $\epsilon \rightarrow \epsilon'$, where $\rho(\epsilon) = (1-\epsilon)\rho_0 + \epsilon \rho_0^\perp$, with $\rho_0$ the ideal single photon state. The proposed scheme is shown in Fig.~\ref{fig:distillation}, which distills one photon from three photons upon valid post-selection, on average requiring three iterations to run, consuming $\sim$9 photons in total.
 
 The intuition behind the scheme is fairly simple; the rate at which successful post-selections occur depends strongly on the initial state. If three identical photons are present, they are post-selected on at a rate $1/3$. If however there is an error, such as $|\psi_0, \psi_0^\perp, \psi_0\rangle$, then this occurs only at rate $1/9$. This means that errors are naturally filtered out, and reduced in magnitude by about a factor of 3 per iteration, $\epsilon \rightarrow \epsilon/3$.
 
 There is an additional consideration, and that is the robustness of the scheme to other errors. Whilst it is shown in Ref. \citenum{marshall_distillation_2022} there is some built-in robustness to certain detection errors and control errors, there is no real protection from photon losses, which are typically the largest error source. Although this does not particularly pose a threat to reducing the quality of the output photons (since losses can typically be heralded), it does make the scheme less efficient, i.e., in practical settings with losses the circuit will need to be run more times to achieve success. There is therefore a cost-benefit analysis to consider; the benefit of the added number of distillation iterations to reducing distinguishability, Vs the added cost of longer circuits being more likely to fail due to photon losses.

There are several questions one can consider based on this work, such as incorporation to resource state generation circuits (i.e., circuits that are naturally tolerant to such errors), the cost of distillation in the presence of realistic errors (losses, multi-photon etc.), and generalizations of the scheme to higher photon numbers.

 \begin{figure}
     \centering
     \includegraphics[width=0.4\columnwidth]{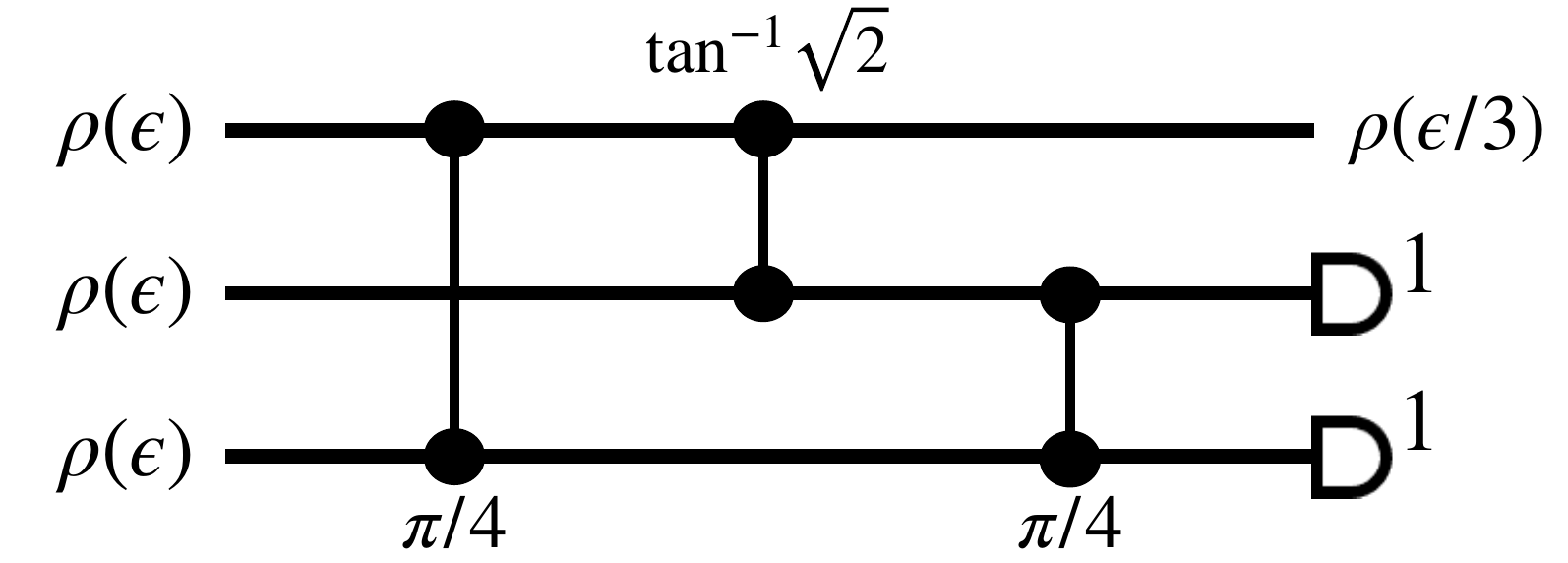}
     \caption{3 photon distillation scheme of Ref.  \citenum{marshall_distillation_2022}. As indicated, the output is post-selected upon measurement result $(1,1)$. The angles represent the transmission angle in the beamsplitters (see Ref. \citenum{marshall_distillation_2022} for a description).}
     \label{fig:distillation}
 \end{figure}

\subsection{Impact of errors on notable circuits}\label{sec:distinguishability impact}

In this section, we will be concerned with operations on photonic dual-rail states and their interactions with distinguishability errors. A dual-rail qubit has state space spanned by $\ket{\mathbf{0}} = \ket{1,0}, \ket{\mathbf{1}} = \ket{0,1}$, a subspace of the Fock space of $1$ photon in $2$ modes. We write $\ket{\pm} = \frac{1}{\sqrt{2}}\ket{1,0}\pm\ket{0,1}$. 
We will use the standard notation $X,Y,Z$ for the dual-rail single-qubit Pauli operators; these are linear optical unitaries, where $Z$ applies a phase to the second mode, $X$ swaps the two modes, and $Y=iXZ$. 

Since linear optics is not universal, one must obtain universality via other means (see the discussions in Sec.~\ref{sec:designs supersection}). 
Browne and Rudolph \cite{browne_resource-efficient_2005} introduced Type I and Type II \emph{fusion} operations, nondeterministic post-selection operations depicted in Figure~\ref{fig:fusion}, for this purpose. 
There are many variants of fusion, notably boosted fusion (see e.g. Refs.~\citenum{ewert_3_2014, bartolucci2021creation}), which increases the success rate at the cost of additional resources. 
In the present work, we will not consider these generalizations, instead analyzing the standard fusion operations and the generalization of Type II fusion to an $n$-GHZ state projection.

\subsubsection{Distinguishability error analysis for (generalized) fusion}\label{sec:fusion_dist_statements}

In this section, we give a detailed study of Type I and II fusion operations and their generalization to GHZ state analyzers in the presence of distinguishability. 
We begin by analyzing the circuit in Figure~\ref{fig:fusion}(a), 
state its consequences for Type I fusion, 
then proceed to the GHZ state analyzer introduced in Refs. \citenum{gimeno2016towards, bartolucci2021creation} as part of a GHZ state generation protocol. 
This GHZ state analyzer is given in Figure~\ref{fig:ghz_analyzer}. 
These results are mostly straightforward calculations, with the most difficult part being the notation. 
We choose to state these results here in detail, with all signs worked out, so that they can be easily found in the literature. 
The proofs are sketched in Appendix~\ref{sec:rsg proofs}. 

Following the model of distinguishability presented in \eqref{eq:dist error model} and used in Sec.~\ref{sec:distinguishability}, we assume that the $i$th photon, $i\geq 1$, has internal state 
    $\rho_i(\epsilon_i) = (1-\epsilon_i) |\psi_0\rangle \langle \psi_0| + \epsilon_i \ketbra{\psi_i}$,
where $\braket{\psi_0}{\psi_i}=0$ for $i\geq 1$, 
the ideal state $\ket{\psi_0}$ is shared by all photons, and $\epsilon$ and $\ket{\psi_i}$ may vary. 
Using this mixed state framework, it suffices to consider photons that are either \emph{ideal} (in state $\ket{\psi_0}$) or \emph{fully distinguishable} from the ideal state (in some state $\ket{\psi_i}$ orthogonal to $\ket{\psi_0}$), then take linear combinations. 
As a consequence, if we begin with such photons and apply linear optical circuits (without measurement), we may assume their internal states are pure states.

 \begin{figure}
     \centering
     \includegraphics[width=0.8\columnwidth]{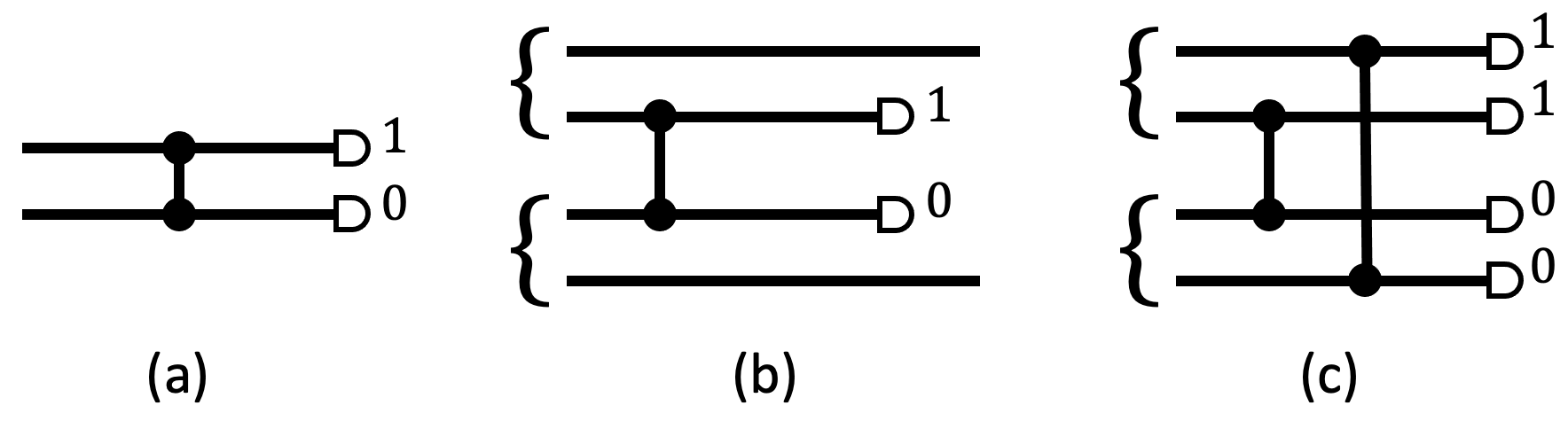}
     \caption{(a) Generalized $X$ measurement. A beamsplitter with transfer matrix $H$ is applied to two modes (generally part of a larger circuit, not pictured). PNRD is then performed on both modes, and we post-select upon measuring exactly $1$ photon total. The figure shows the example with measurement pattern $(1,0)$; pattern $(0,1)$ is also permissible. 
     Often, future steps in a computation will be altered depending on the pattern obtained here. 
     (b) Type $I$ fusion. Dual-rail states are input in each pair of modes (indicated by curly braces). We then apply the generalized $X$ measurement circuit on the two middle modes. Upon successful post-selection, the two unmeasured modes are treated as a single dual-rail qubit. 
     (c) Type $II$ fusion. We input $2$ dual-rail states, then perform $2$ generalized $X$ measurements. 
     If successful, this circuit is equivalent to measuring the pair of commuting two-qubit dual-rail observables $X\otimes X$ and $Z\otimes Z$, with the eigenvalues determined by the obtained measurement patterns. 
     This is discussed in greater generality below.
     }
     \label{fig:fusion}
 \end{figure}

To begin, we consider the circuit in Figure~\ref{fig:fusion}(a), in which two modes undergo a Hadamard beamsplitter and are then measured, and we post-select for measurement of exactly one photon. 
We call this a \emph{generalized $X$ measurement}, 
as the same operation applied to a dual-rail qubit is precisely a measurement in the $X$ basis. 
This operation \emph{heralds success} if one of the desired patterns, $(1,0)$ or $(0,1)$, is measured (projecting to $|+\rangle$ or $|-\rangle$ for a dual-rail qubit). 

\begin{lemma}\label{lemma:x measurement}
Consider a state $\ket{\phi}$ on $M\geq 2$ modes, where two modes are given as input to a generalized $X$ measurement. Suppose that this generalized $X$ measurement heralds success. 
Only the terms of $\ket{\phi}$ involving exactly one photon in the input modes will affect the output state. 
\end{lemma}

Given the assumptions of Lemma~\ref{lemma:x measurement}, 
it suffices to consider only those terms of $\ket{\phi}$ with state $\ket{1,0}$ or $\ket{0,1}$ in the input modes (which we index by $0$ and $1$ for convenience). 
In principle, there may be many such terms, of the form $\cre{0}(\xi)\vac$ or $\cre{1}(\xi)\vac$ with various possible internal states $\ket{\xi}$. 
(Here $\cre{i}(\xi)$ is the creation operator for a photon in mode $i$ with internal state $\ket{\xi}$.) 
We will focus on the special case in which there is at most one such term for each of $\cre{0}, \cre{1}$. 
By using linearity and appropriate simplifications, this will be sufficient for our purposes. 

\begin{proposition}\label{prop:x measurement}
With the assumptions of Lemma~\ref{lemma:x measurement}, 
further assume as above that the relevant terms of $\ket{\phi}$ (as in the Lemma) are expressed as $c_{10}\ket{\phi_{10}}\cre{0}(\xi)\vac + c_{01}\ket{\phi_{01}}\cre{1}(\xi')\vac$, 
where $\ket{\phi_{10}}, \ket{\phi_{01}}$ describe the state on the non-input modes, the input modes are labeled by $0$ and $1$, and $\ket{\xi}, \ket{\xi'}$ are the internal states of the appropriate photons (pure, for convenience). 
Success is heralded with probability $|c_{10}|^2 + |c_{01}|^2$.
The output state is (up to normalization)
\begin{align}\label{eq:dist x measurement general}
    \frac{1}{2}\left( |c_{10}|^2\ketbra{\phi_{10}} + |c_{01}|^2\ketbra{\phi_{01}}\right)
    \pm \textnormal{Re}
    \left( \braket{\xi'}{\xi}c_{10}\overline{c_{01}}\ketbra{\phi_{10}}{\phi_{01}}
    \right),
\end{align}
where $\textnormal{Re}(\rho) = \frac{1}{2}(\rho + \rho^\dagger)$ and the sign is $+1$ if $(1,0)$ is measured and $-1$ if $(0,1)$ is measured. 
In particular, if the photons corresponding to the terms $\ket{1,0}$ and $\ket{0,1}$ are mutually indistinguishable, then the output state is 
the pure state $\frac{1}{\sqrt{2}}\left(c_{10}\ket{\phi_{10}} \pm c_{01}\ket{\phi_{01}} \right)$. 
If the photons are mutually fully distinguishable (and both $c_{10}, c_{01}\neq 0$), then the output state is (up to normalization) the mixed state
    $\frac{1}{2}\left( |c_{10}|^2\ketbra{\phi_{10}} + |c_{01}|^2\ketbra{\phi_{01}}  \right).$
\end{proposition}

In words, we note that an ideal generalized $X$ measurement leads to a superposition of the $\ket{\phi_{10}}$ and $\ket{\phi_{01}}$ terms, and distinguishability instead leads to a mixture of the two. 

This can be applied to Type I fusion, which is simply a generalized $X$ measurement on a pair of modes, each making up half of a dual-rail qubit (Fig.~\ref{fig:fusion}b): 
\begin{proposition}\label{prop:type i}
Consider a tensor product state $\ket{\phi^{(L)}}\otimes \ket{\phi^{(R)}}$, where two modes of each (corresponding to a pair of dual-rail qubits) are given as input to Type I fusion. 
Assume that the fusion (equivalently, the generalized $X$ measurement) heralds success, and further that the four input modes do in fact correspond to a pair of dual-rail qubits. 
Similar to above, the terms affecting a successful Type I fusion have the form (with modes appropriately permuted)
    $c_{10}^{(L)}c_{10}^{(R)}\ket{\phi_{10}^{(L)}}\ket{\phi_{10}^{(R)}}\ket{1,0,1,0} + c_{01}^{(L)}c_{01}^{(R)}\ket{\phi_{01}^{(L)}}\ket{\phi_{01}^{(R)}}\ket{0,1,0,1},$
where we again assume that the internal states of the input photons are pure states. 
If the relevant input photons are all ideal (mutually indistinguishable), then the output state is (up to normalization)
\begin{align}
    c_{10}^{(L)}c_{10}^{(R)}\ket{\phi_{10}^{(L)}}\ket{\phi_{10}^{(R)}}\ket{1,0} \pm c_{01}^{(L)}c_{01}^{(R)}\ket{\phi_{01}^{(L)}}\ket{\phi_{01}^{(R)}}\ket{0,1},
\end{align} 
where the sign is $+1$ for measurement pattern $(1,0)$ and $-1$ for $(0,1)$. 
If the photons are fully distinguishable, then for either measurement outcome, the output state is the even mixture of the two unentangled states 
\begin{align}
    c_{10}^{(L)}c_{10}^{(R)} \ket{\phi_{10}^{(L)}}\ket{\phi_{10}^{(R)}}\ket{1,0},\, c_{01}^{(L)}c_{01}^{(R)}\ket{\phi_{01}^{(L)}}\ket{\phi_{01}^{(R)}}\ket{0,1,0,1}.
\end{align}
\end{proposition}

We now apply Proposition~\ref{prop:x measurement} to the study of the generalized GHZ state analyzer in Figure~\ref{fig:ghz_analyzer}. 
The $n$-GHZ analyzer involves Hadamard beamsplitters on the mode pairs $\mc{P} = \{(1,2), (3,4), \dots, (2n-3, 2n-2), (2n-1, 0)\}$. (Note that the Hadamard is not symmetric between the two modes: our convention is that the second mode in the pair receives the nontrivial phase. In particular, note that with this convention it is mode $0$, \emph{not} $2n-1$, that receives the phase.) This is followed with post-selection as described in the figure. 

 \begin{figure}
     \centering
     \includegraphics[width=0.25\columnwidth]{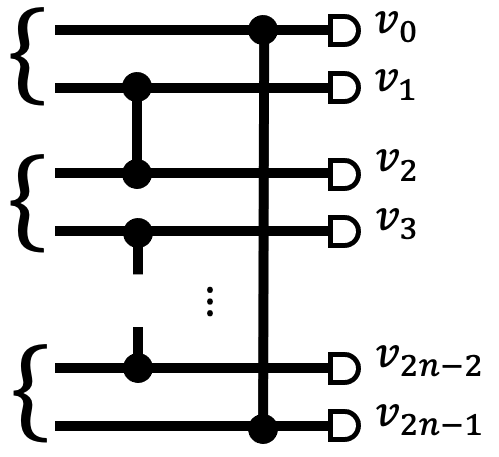}
     \caption{The $n$-GHZ state analyzer: Type II fusion is the special case with $n=2$. 
     We apply a Hadamard beamsplitter for each $p\in\mc{P}$, then post-select for patterns $m = (m_0, \dots, m_{2n-1})$ with all $m_{2i+1} + m_{2i+2} = 1$ (including $m_{0} + m_{2n-1} = 1$, due to the convention $m_{2n} = m_0$). 
     This operation may be viewed as $n$ generalized $X$ measurements between the modes of $n$ dual-rail input states, where the $n$-GHZ state analyzer is considered successful if and only if all the generalized $X$ measurements are.
     }
     \label{fig:ghz_analyzer}
 \end{figure}

We begin with the following lemma, which is immediate from the definition of the $n$-GHZ state analyzer. 
We use the notation $\vec{m} = (m_0, \dots, m_{2n-1})$, $\ket{\vec{m}} = \ket{m_0, \dots, m_{2n-1}}$. 
For convenience, we use the convention $m_{2n} = m_0$ (to match with the mode pairs $\mc{P}$ above). 
We will make frequent reference to the $2n$-tuples 
$(1,0)^n = (1,0,1,0, \dots, 1,0), (0,1)^n = (0,1,0,1,\dots, 0,1)$. 

In the following lemma, the internal states of the input photons are suppressed from the notation. 
These states do not affect the content or the proof. 

\begin{lemma}\label{lemma:ghz filtering}
Suppose the final $2n$ modes of the state $\ket{\phi}$ 
are given as input to the $n$-GHZ state analyzer, and that the analyzer heralds success. 
Given this successful heralding, only terms of the following form may affect the output state: 
\begin{align}\label{ghzfusion:projected}
    \sum_{\vec{m}\in\mathcal{S}} c_{\vec{m}}\ket{\phi_{\vec{m}}}\ket{\vec{m}},
\end{align}
where $\mathcal{S} = \{\vec{m} : c_{\vec{m}}\neq 0, m_{2i+1} + m_{2i + 2} = 1\textnormal{ for all }i\}$. 
If, for all $\vec{m}\in\mathcal{S}$ and all $0\leq i <n$, we have $(m_{2i}, m_{2i+1})\not\in\{(1,1), (0,0)\}$, then only the two terms $\vec{m}=(1,0)^n, (0,1)^n$ may affect the measurement (and only if both have nonzero coefficient). 
\end{lemma}

In other words, many patterns that are not legitimate dual-rail qubit patterns, such as those with two photons in a single mode, are filtered out. 
If we are able to rule out the additional patterns $(0,0)$ and $(1,1)$, which is natural in the $n$-GHZ state generation below, then we in fact filter out all but two terms. 
Once we make this reduction, we will want to more carefully consider the internal states of each photon. Assume that the photon in mode $2i+1$ of $(0,1)^n$ (respectively mode $2i+2$ of $(1,0)^n$) has internal state $\ket{\xi_i}$ (respectively $\ket{\xi_i'}$). We then express the relevant terms as
\begin{align}\label{eq:reduced form}
    c_{01}\ket{\phi_{01}}\ket{\xi_0, \dots, \xi_{n-1}}\ket{0,1}^{\otimes n} + c_{10}\ket{\phi_{10}}\ket{\xi_0', \dots, \xi_{n-1}'}\ket{1,0}^{\otimes n},
\end{align}
where $c_{01} = c_{(0,1)^n}$, etc. 
Note, as discussed before Proposition~\ref{prop:x measurement}, that it is an \emph{assumption} that we can express the relevant internal states as pure states. This assumption will be natural in practice. 

\begin{theorem}\label{thm:ghz analyzer}
Suppose the final $2n$ modes of a state 
are given as input to the $n$-GHZ state analyzer, and that the analyzer heralds success. 
As above, assume that for all $\vec{m}$ with coefficient $c_{\vec{m}}\neq 0$ and all $0\leq i <n$, we have $(m_{2i}, m_{2i+1})\not\in\{(1,1), (0,0)\}$, 
and that we can express the relevant terms of the input state in the form \eqref{eq:reduced form}. 
We have the following: 
\begin{enumerate}
\item  The $n$-GHZ state analyzer projects onto the $+1$ eigenspaces of the dual-rail qubit observables $Z_j Z_{j+1}$ for $0\leq j\leq n-1$. 
\item Let $s_{\textnormal{odd}}$ be the number of photons measured by the $n$-GHZ state analyzer in odd-indexed modes $1, \dots, 2n-1$. Up to normalization, the output of the $n$-GHZ state analyzer is
\begin{align}
    \dfrac{1}{2}\left( |c_{10}|^2\ketbra{\phi_{10}} + |c_{01}|^2\ketbra{\phi_{01}}\right)
     + (-1)^{s_\textnormal{odd}} \textnormal{Re}\left( \left(\prod_i \braket{\xi_i}{\xi_i'}\right) \ketbra{\phi_{10}}{\phi_{01}}
     \right).
\end{align}
\item Suppose that all $\braket{\xi_i}{\xi_i'} = 1$. 
Then the $n$-GHZ state analyzer measures the dual-rail qubit observable $X_0\cdots X_{n-1}$ with eigenvalue $(-1)^{s_\textnormal{odd}}$ and yields (up to normalization) the output state 
\begin{align}\label{eq:ghz obtained}
    \frac{1}{\sqrt{2}}\left(c_{10}\ket{\phi_{10}} +(-1)^{s_{\textnormal{odd}}} c_{01}\ket{\phi_{01}}\right).
\end{align}

\item Suppose that for some $i$, $\braket{\xi_i}{\xi_i'} = 0$. 
Then the measurement value of the observable $X_0\cdots X_{n-1}$ is fully randomized. The GHZ analyzer outputs, up to normalization, the mixed state
\begin{align}
    \frac{1}{2}\left(|c_{10}|^2\ketbra{\phi_{10}} + |c_{01}|^2\ketbra{\phi_{01}}\right).
\end{align}
\end{enumerate}
\end{theorem}

\subsubsection{Distinguishability error analysis for GHZ state generation}\label{sec:dist ghz}

We now analyze a family of protocols for the generation of $n$-GHZ states from single photons. The protocol in the case $n=3$ is due to Ref. \citenum{varnava2008good}, and the generalization was given in Ref. \citenum{gimeno2016towards}. 
We elaborate upon the treatment in Ref. \citenum{bartolucci2021creation}. 
We also note that the $n=2$ (Bell pair) case was analyzed in the presence of distinguishability errors by Sparrow \cite{sparrowQuantumInterferenceUniversal2017}. 

\begin{figure}
 \centering
 \includegraphics[width=0.4\columnwidth]{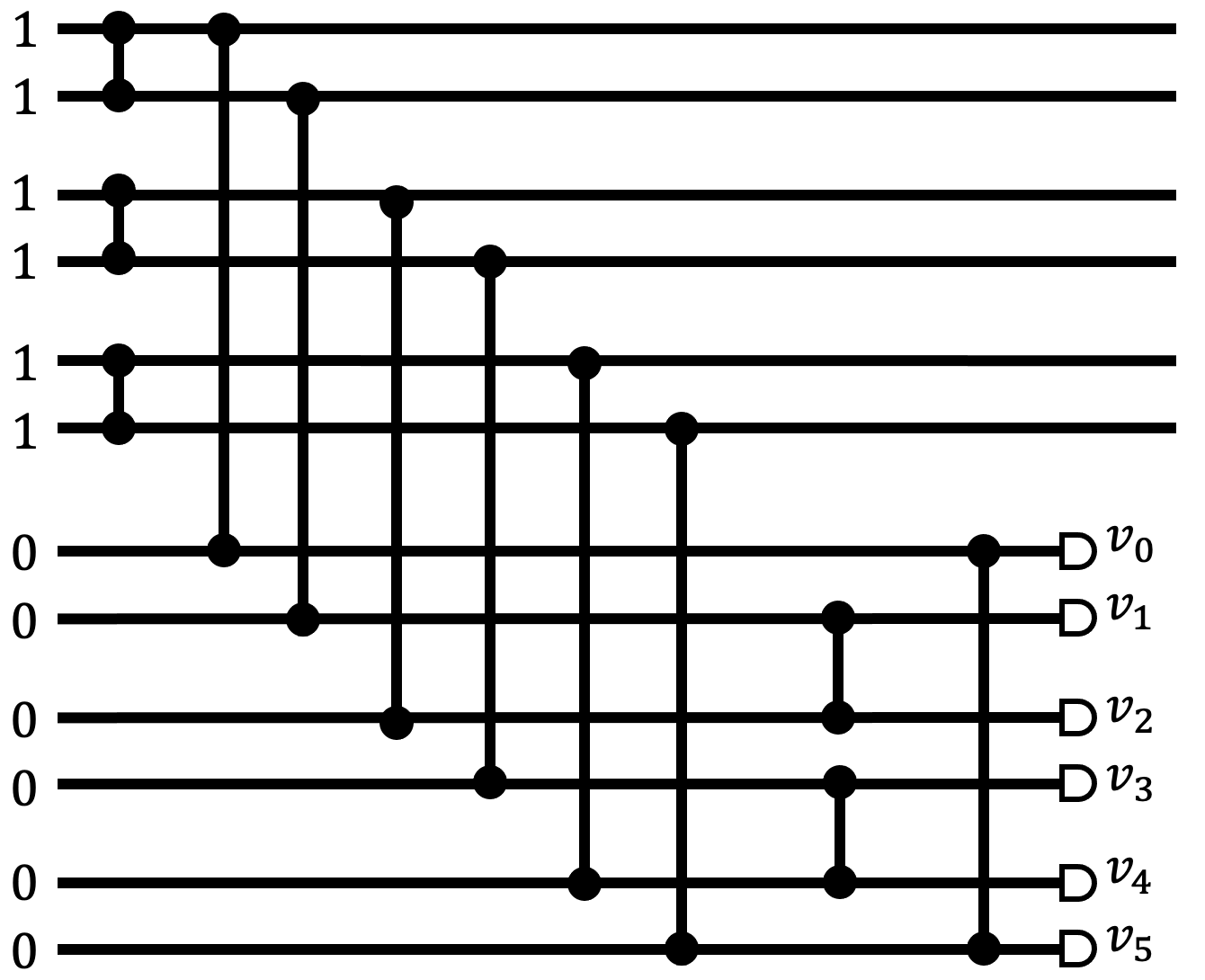}
 \caption{Preparation of an $n$-GHZ state for $n=3$. The circuit is divided into two halves, each with $6=2n$ modes. 
 There are $3$ main steps. 
 First, in the top half of the circuit, we input a photon in each mode, then perform an $H$ beamsplitter between each pair of modes. 
 In the ideal case, this prepares a $2$-photon N00N state in each of the $n=3$ mode pairs. 
 Second, for all $0\leq i <6 (=2n)$, we perform an $H$ beamsplitter between the $i$th mode in the top half and the $i$th mode in the bottom half. 
 Finally, we perform a $3$-GHZ state analyzer on the bottom half (see Figure~\ref{fig:ghz_analyzer}), corresponding to $3$ generalized $X$ measurements as pictured. 
 We say the measurement has succeeded if and only if each generalized $X$ measurement succeeds; in other words, we post-select for a measurement pattern $v = (v_0, \dots, v_5)$ with $v_1 + v_2 = v_3 + v_4 = v_0 + v_5 = 1$. 
 }
 \label{fig:3ghz}
\end{figure}

The $n$-GHZ protocol is as follows. See Figure~\ref{fig:3ghz} for the case $n=3$. We fix an integer $n\geq 3$ and consider a linear optical circuit involving $2n$ photons in $4n$ modes, indexed $0, 1, \dots, 4n-1$. 
We will only require $50:50$ beamsplitters (we use the Hadamard transfer matrix $H$, although other versions would be sufficient) and PNRD. 
\begin{protocol}\label{protocol:ghz}
\begin{enumerate}
    \item Input the Fock state $\ket{1}^{\otimes 2n} \otimes \ket{0}^{\otimes 2n}$. \label{ghz:input}
    \item Apply $H$ on each pair $(0,1), (2,3), \dots, (2n-2, 2n-1)$. \label{ghz:noon}
    \item Apply $H$ on each pair $(0, 2n), (1, 2n+1), \dots, (2n-1, 4n-1)$. \label{ghz:splitter}
    \item Apply the $n$-GHZ state analyzer, including post-selection, on modes $(2n, \dots, 4n-1)$. 
    Proceed if and only if the analyzer heralds success. 
    \item If the GHZ analyzer reports a measurement of $X_{2n}\cdots X_{4n-1} = -1$, then apply a phase of $-1$ to mode $1$. \label{ghz:correct}
\end{enumerate}
\end{protocol}
We have
\begin{lemma}\label{lemma:ghz gen}
Suppose the $n$-GHZ state generation protocol heralds success.  
If all input photons have ideal internal state $\eta_0 = \ketbra{\psi_0}$, the output state is 
$\ket{B_+}= \frac{1}{\sqrt{2}}(\ket{1,0}^{\otimes n} + \ket{0,1}^{\otimes n})$, the desired ideal $n$-GHZ state (with all photons in the ideal internal state). 
\end{lemma}

In Appendix~\ref{sec:rsg proofs}, we sketch the proof of Lemma~\ref{lemma:ghz gen} and the potential errors in the distinguishable case.

Distinguishability (or loss) in the input photons may result in output states that are not dual-rail qubit states, due to mode pairs containing no photons or more than $1$ photon. 
However, it is likely that these illegitimate patterns will be filtered out during later steps in the computation, as often occurs when applying an $n$-GHZ state analyzer (see Lemma~\ref{lemma:ghz filtering}). 
Thus we are often interested in the case where we consider only the \emph{dual rail projection}, meaning we project to the space involving only terms that can arise in a legitimate dual-rail state. 

\begin{remark}\label{rem:loss filtering}
Note that photon loss, sampled as in \eqref{eq:mixed loss}, will always be filtered out by the dual-rail projection, as loss reduces the total photon number. 
This is the motivation for considering only distinguishability errors here. 
For more complex circuits, higher-order error models, or protocols where we are particularly concerned with the heralding rates, it is more relevant to incorporate photon loss as well. 
\end{remark}

In the following result, we consider the performance of $n$-GHZ generation under low-order errors, where all but a few input photons have ideal internal state $\eta_0 = \ketbra{\psi_0}$. 
We first consider a first-order error in which exactly one input photon has internal state $\eta_1 = \ketbra{\psi_1}$ with $\Tr(\eta_0 \eta_1) = 0$. 
We further assume that errors in modes $2i, 2i+1$ are equally likely; 
thus, we will study the performance of the $n$-GHZ state generation protocol where the $i$th input pair is replaced with
\begin{align}\label{eq:first order input}
    \frac{1}{2}\left(\cre{2i}(\eta_0)\cre{2i+1}(\eta_1)\vac\bvac a_{2i+1}(\eta_1)a_{2i}(\eta_0) + \cre{2i}(\eta_1)\cre{2i+1}(\eta_0)\vac\bvac a_{2i+1}(\eta_0) a_{2i}(\eta_1)\right). 
\end{align}
This is an even mixture of the cases in which the photon in mode $2i$ (respectively $2i+1$) is distinguishable. 
We will also consider a second-order \emph{pair error}, in which the input photons in \emph{both} modes $2i$ and $2i+1$ are distinguishable with the same internal state $\eta_1 = \ketbra{\psi_1}$ with $\Tr(\eta_0 \eta_1) = 0$. 

We note that in Ref. \citenum{sparrowQuantumInterferenceUniversal2017}, the output state of the Bell state generator was very nicely calculated by entirely tracing out the internal states of all photons. 
Here we choose not to trace out the internal degrees of freedom entirely, instead giving a framework in which distinguishability errors on the input photons may be mapped to a combination of Pauli and distinguishability errors on the output state. 
We briefly introduce some notation for this purpose. 
For an operator $A$, let $P_A(\gamma) = \frac{1}{2}(\gamma + A \gamma A^\dagger)$. 
Also let $D_i$ be an operator that acts on modes $2i, 2i+1$ by $\creta{2i}{0}\vac\mapsto \creta{2i}{1}\vac$, $\creta{2i+1}{0}\vac\mapsto \creta{2i+1}{1}\vac$. 
This operator is simply for convenience of notation, and allows us to write partially-distinguishable mixed states by starting with the ideal state $\ketbra{B_+}$ and applying appropriate operators. 

\begin{theorem}\label{thm:ghz gen}
Suppose the $n$-GHZ state generation protocol heralds success with measurement pattern $\vec{m} = (m_0, \dots, m_{2n-1})$. 
Assume that all input photons except for those in input modes $2i, 2i+1$ are ideal, with internal state $\eta_0 = \ketbra{\psi_0}$. 
\begin{enumerate}
    \item Suppose that input photons $2i$, $2i+1$ are in state \eqref{eq:first order input} above. 
    After dual-rail projection, the output state is
    \begin{align}
    \frac{1}{2}P_{X_i}\left( P_{Z_i}(\ketbra{B_+}) + D_i(\ketbra{B_+})  \right).
    \end{align}
    \item Suppose we have a pair error on pair $i$, so that input photons $2i, 2i+1$ are in the same internal state $\eta_1$, fully distinguishable from the ideal internal state $\eta_0$. 
    Then the output state is $P_{Z_i}(D_i(\ketbra{B_+}))$, 
    with no dual-rail projection required. 
    
\end{enumerate}
\end{theorem}

\begin{remark}\label{rem:stabilizer sims}
We note that these results suggest a paradigm for computing the effect of low-order distinguishability errors on much larger circuits. 
Suppose we have a large circuit whose subroutines are $n$-GHZ state generation, single-qubit Clifford operations, Type I fusions, and $k$-GHZ state analyzers (for varying $n,k$). 
Also suppose that the distinguishability error rates are sufficiently low that low-order error approximations are reasonable. 
Theorem~\ref{thm:ghz gen} shows that, under these circumstances, erroneous $n$-GHZ state generation may be simulated by taking the ideal state and probabilistically applying Pauli errors and an operator $D_i$ that converts ideal photons to fully distinguishable ones. 
This sampled state may then be fed into the later unitary operations, fusion gates, etc.; it is straightforward to see how Pauli errors propagate through the circuit, as with standard Pauli frame calculations, and the results of Sec.~\ref{sec:fusion_dist_statements} show how distinguishability errors interact with each operation. 
This implies the possibility of a generalized stabilizer-type simulation, in which one tracks both the stabilizers of the simulated state and, for each photon, a minimal amount of information about its internal state. 
This, of course, will be simplest if one uses an appropriately simple distinguishability error model, such as those discussed in Sec.~\ref{sec:numerics} below. 
Further, the assumption that the errors are low-order is essential: otherwise, multiple interacting errors could "cancel," resulting in nontrivial contributions from 
terms that are normally "filtered out" by subsequent operations (as discussed in Lemma~\ref{lemma:ghz filtering} and below). 
\end{remark}

\subsubsection{Numerical distinguishability error analysis for Bell state generation}\label{sec:numerics}
In the previous section, we characterized the behavior of the $n$-GHZ state generation protocol, Protocol~\ref{protocol:ghz}, under low-order distinguishability errors. 
We now use our simulation tools to numerically study the performance of the Bell state generation protocol (the special case $n=2$) under all orders of distinguishability errors, using various performance metrics and error models. 
We note that Sparrow \cite{sparrowQuantumInterferenceUniversal2017} has explicitly calculated the output state of this protocol (after dropping non-dual-rail terms and tracing out all internal degrees of freedom) as a function of each input photon's internal state; further, Shaw et al. \cite{shaw2023errors} have numerically compared the fidelity of this protocol with other BSG protocols. 
Here, our goal is to demonstrate the utility of our simulation paradigm for analyzing such protocols, and to examine the interactions between different error models and performance metrics. 
These techniques can be extended to larger protocols (especially via the simulation paradigm discussed in Sec.~\ref{sec:coherent} below) and can easily incorporate additional types of errors. 
(Note that, as discussed in Remark~\ref{rem:loss filtering}, we neglect loss errors in this section as they will be filtered out by all the metrics we consider here.)

We consider three different models for distinguishability errors. 
In the first two, as in \eqref{eq:dist error model} above, the $i$th photon's internal state is independent and is modeled as $\rho(\epsilon) = (1-\epsilon)\ketbra{\psi_0} + \epsilon\ketbra{\psi_i}$, where $\braket{\psi_0}{\psi_i} = 0$ for $i>0$. 
In the Orthogonal Bad Bits (OBB) model \cite{sparrowQuantumInterferenceUniversal2017}, we have $\braket{\psi_i}{\psi_j}=0$ for $i\neq j$. 
In the Same Bad Bits (SBB) model, we have only one error state: $\ket{\psi_1}=\ket{\psi_2}=\dots$. 
Finally, in the Orthogonal Bad Pairs (OBP) model (corresponding to the pair errors studied in Theorem~\ref{thm:ghz gen}), 
pairs of mutually indistinguishable photons are emitted. Each pair is independent of the others, with a unique error state. In other words, the $i$th \emph{pair} has internal state of the form $(1-\epsilon')\ket{\psi_0\psi_0}\bra{\psi_0\psi_0} + \epsilon' \ket{\psi_i\psi_i}\bra{\psi_i\psi_i}$, where $\braket{\psi_i}{\psi_j}=0$ for $i\neq j$. 
These mutually indistinguishable pairs are routed into adjacent modes $(0,1)$ or $(2,3)$ of the BSG circuit. 
In order to compare to the other models, we will take $\epsilon' = 1-(1-\epsilon)^2 = \epsilon(2-\epsilon)$, so that in all three models, the probability of an ideal pair is $(1-\epsilon)^2$. 

We also consider four different metrics for benchmarking the quality of the four-mode state $\eta = \sum_j p_j \ketbra{\eta_j}$ constructed by the BSG protocol (ideally a two-qubit dual-rail state). 
The first is the fidelity with the desired ideal Bell pair $\ket{B_{++}}$. (Here we use the notation that $\ket{B_{s_1 s_2}}$ is the ideal Bell state with $(X\otimes X) \ket{B_{s_1 s_2}} = s_1 \ket{B_{s_1 s_2}}$, $(Z\otimes Z) \ket{B_{s_1 s_2}} = s_2 \ket{B_{s_1 s_2}}$.) 
The second is the \emph{post-selected fidelity}: we drop all terms in the output state that do not correspond to a dual-rail state, renormalize, then calculate the fidelity with the ideal state $\ket{B_{++}}$. This is motivated by the fact that fusion and similar circuits tend to filter out non-dual-rail terms such as $\ket{11}, \ket{20}$, etc. (recall Lemma~\ref{lemma:ghz filtering}). 
For the final two metrics, we take this observation to its logical conclusion, asking the following question: if the state $\eta$ is used in a fusion-based circuit, what is the expected \emph{number of stabilizer errors} on the resulting state? 
To quantify this, we perform \emph{post-processing fusions} on the state $\eta$: we fuse the first two modes (first qubit) of $\eta$ with an ideal Bell pair, and the final two modes (second qubit) of $\eta$ with a distinct ideal Bell pair, post-selecting for successful heralding. 
(This is done without additional error, as we are attempting to benchmark the output of the BSG, not the quality of these subsequent fusion operations.) 
Successful post-processing fusions may be viewed as mapping the four-mode state $\eta$ to the four-mode state on the unmeasured output modes; in the absence of distinguishable photons input to the post-processing fusions, this map is in fact a non-destructive Bell measurement. 
Further, since the surviving qubits come from the ideal Bell pairs, the output is always a dual-rail qubit state. 
Then it makes sense to take this state, measure its fidelity with each of the four Bell states, and use this to calculate the expected number of errors in the stabilizers $\{X\otimes X, Z\otimes Z\}$: 
\begin{align}
    E(NS) = 0 \bra{B_{++}}\eta\ket{B_{++}} + 1 (\bra{B_{+-}}\eta\ket{B_{+-}} + \bra{B_{-+}}\eta\ket{B_{-+}}) + 2\bra{B_{--}}\eta\ket{B_{--}}. 
\end{align}
For comparison with the previous two fidelity metrics, we also calculate the \emph{fusion fidelity}, the fidelity of the state obtained after the post-processing fusions. 

Now, let $\rho$ be the input state to the BSG protocol with a chosen error model. Recalling the discussion above, we will break $\rho$ into a sum 
$\rho = \sum_{c\in\mathcal{C}} P(c) \rho_c$, 
where each photon in $\rho_c$ has some internal state $\ket{\psi_i}$, either ideal or fully distinguishable according to the appropriate error model, and 
$\mathcal{C}$ is the set indexing all such states. 
Now let $X$ be a random variable on the space of four-mode photonic density matrices of the appropriate form: we assume $X$ is obtained by applying some linear optical circuit, measuring a subset of modes and post-selecting for a subset of \emph{heralding} outcomes, applying linear optical corrections as needed, then applying some random variable $X'$ to the resulting output state. 
We let $H_X$ be the event that the heralding conditions for $X$ are met. 
We may calculate the expected value of $X$ given successful heralding as follows: 
\begin{align}
    E(X|H_X) = \dfrac{\sum_{c\in\mc{C}} P(c) P(H_X|c) E(X|c\cap H_X)}{\sum_{c\in\mc{C}} P(c) P(H_X|c)}.
\end{align}
Note that for $X$ corresponding to the post-selected fidelity and number of stabilizer errors, the heralding probabilities $P(H_X|c)$ take into account the reduced probabilities due to the post-selection and post-processing fusions respectively. 
In Appendix~\ref{app:numerics}, we explicitly calculate this expected value, for each of the error models and random variables given above. 
To do this, we consider each $c\in\mc{C}$ and directly compute the appropriate heralding probabilities $P(H_X|c)$ and expected values $E(X|c\cap H_X)$ corresponding to the input state $\rho_c$. 
Note that the only appearance of the distinguishability error rate $\epsilon$ is in $P(c)$, which may be explicitly computed for each $c\in\mc{C}$ (depending on the choice of error model). 
Thus we do not need to sample different values of $\epsilon$, rather directly computing $E(X|H_X)$ as a function of $\epsilon$. 
We refer to Appendix~\ref{app:numerics} for the exact results and Figure~\ref{fig:bsg_numerics} for comparisons. 

 \begin{figure}
     \centering
     \includegraphics[width=0.9\columnwidth]{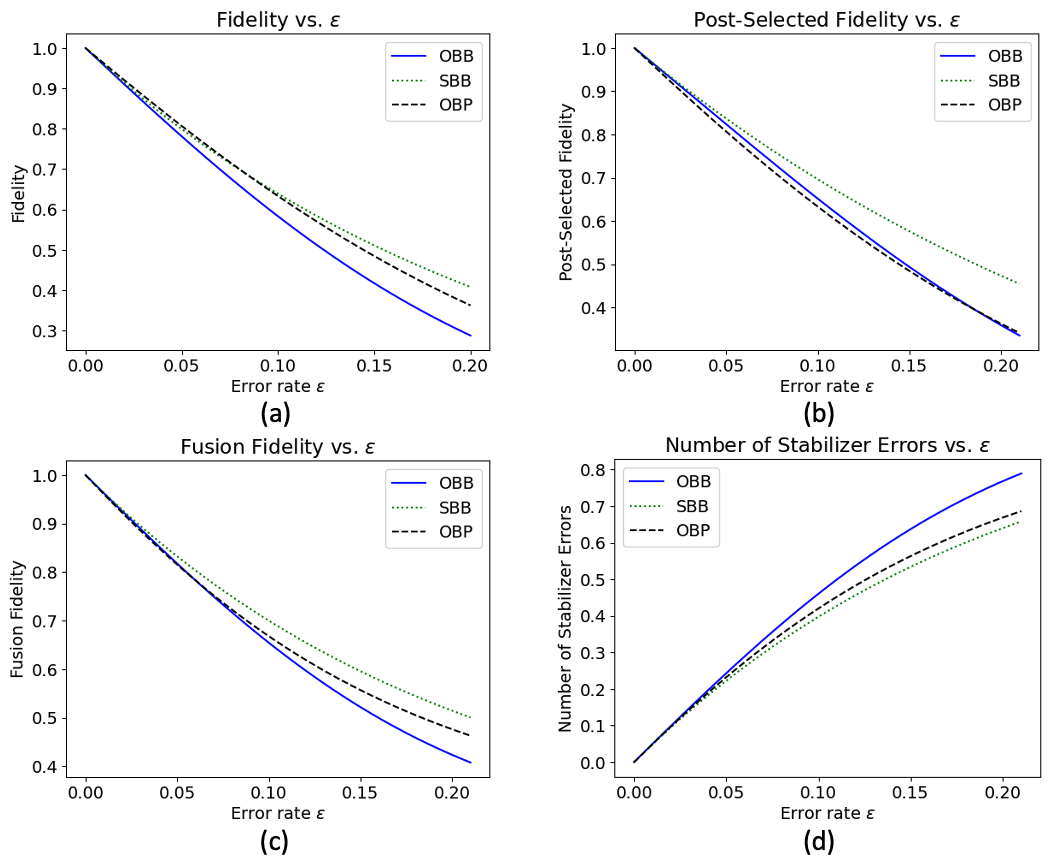}
     \caption{Plot of the BSG circuit's performance under various metrics, comparing 3 different error models.  
     }
     \label{fig:bsg_numerics}
 \end{figure}

We make several observations regarding the numerics in Figure~\ref{fig:bsg_numerics}. 
First, we see that by all metrics, OBB errors lead to worse state generation than SBB errors. 
This fits the intuition that more distinguishability, even between photons that are already non-ideal, reduces the overall performance of photonic circuits. 
The comparison between the OBP and other models, however, is less clear-cut. 
The OBP model leads to higher standard fidelity than both the OBB and SBB models for small $\epsilon$, $0< \epsilon \lessapprox 0.083$. 
However, the OBP model in fact exhibits the \emph{lowest} post-selected fidelity for reasonable values of $\epsilon$, $0 < \epsilon \lessapprox  0.186$; for larger $\epsilon$, OBB is worse than OBP. 
This difference in performance is because the OBP model does not lead to non-dual-rail terms and thus gives the same values for the standard and post-selected fidelity (compare the appropriate functions in Appendix~\ref{app:numerics}).  
Thus, post-selection improves the fidelity in the OBB and SBB models but does not affect the OBP model. 
For the fusion fidelity, we see similar behavior as for the post-selected fidelity, but the intersection between the OBB and OBP curves occurs much earlier, $\epsilon\approx 0.060$. 
Finally, we note that the OBB model leads to more expected stabilizer errors than the OBP model for all $\epsilon$. 
Thus we see that error models and performance metrics can interact nontrivially, and it is important to consider a variety of metrics when comparing the performance of two or more linear optical protocols.

\section{Coherent-rank framework} \label{sec:coherent}

We summarize the coherent-rank framework introduced in Ref.~\citenum{marshall_simulation_2023}. There are two canonical bases of interest for the separable, infinite-dimensional Hilbert space of a simple harmonic oscillator, \(\mathcal{H}_{\infty} \cong L^{2}(\mathbb{R})\): the Fock basis, generated by the eigenstates of the number operator, \(\mathbb{B} = \{ | n \rangle \}_{n=0}^{\infty}\) where \(\hat{n} | n \rangle = n | n \rangle\). And the coherent state basis, \(\widetilde{\mathbb{B}} = \{ | \alpha \rangle \}_{\alpha \in \mathbb{C}}\), which are eigenstates of the annihilation operator, \(\hat{a} | \alpha \rangle = \alpha | \alpha \rangle\). Recall that coherent states are defined via the action of the displacement operator on the vacuum state $|\alpha\rangle:=\hat{D}(\alpha)|0\rangle = e^{-|\alpha|^{2} / 2} \sum_{n=0}^{\infty} \frac{\alpha^{n}}{\sqrt{n !}}|n\rangle$, where \(\hat{D}(\alpha)=e^{\alpha \hat{a}^{\dagger}-\bar{\alpha} \hat{a}}=e^{-|\alpha|^{2} / 2} e^{\alpha \hat{a}^{\dagger}} e^{-\bar{\alpha} \hat{a}}\) is the displacement unitary with \(\alpha \in \mathbb{C}\).

The two bases interact in an interesting way: apart from the vacuum state \(| 0 \rangle\), no other Fock state is a coherent state and vice-versa. As a result, every coherent state \(| \alpha \rangle, \alpha \neq 0\) is a superposition of Fock states and every Fock state \(| n \rangle, n>0\) is a superposition of coherent states. Given two bases in a Hilbert space, it is natural to seek a unitary (or isometry) that  connects the two bases. However in this case, no such transformation exists. This is because the coherent states form an uncountable, overcomplete basis for \(\mathcal{H}_{\infty}\), with the following resolution of the identity,
$\frac{1}{\pi} \int_{\alpha \in \mathbb{C}} d \alpha|\alpha\rangle\langle\alpha|=\mathbb{I}$.
In contrast, the Fock states form a countable orthonormal basis with,
$\sum\limits_{n=0}^{\infty} | n \rangle \langle  n | = \mathbb{I}$.
The inner product between two coherent states is
$\langle\alpha | \beta\rangle=e^{-\frac{1}{2}|\alpha-\beta|^{2}} e^{-i \operatorname{Im}(\alpha \bar{\beta})}$.
This relation makes it evident that if two coherent states are ``far away'' from each other, namely, \(\left| \alpha - \beta \right| \gg 1\), then they are approximately orthogonal.

The above implies a Fock basis representation of a state, such as the maximum $N$-photon single-mode state $|\psi\rangle=\sum_{n=0}^N c_n |n\rangle$, can be specified exactly over the uncountable degrees of freedom: $|\psi\rangle = \frac{1}{\pi} \int d\alpha \langle \alpha | \psi\rangle |\alpha\rangle $. One observation of Ref.~\citenum{marshall_simulation_2023} is that one can also \textit{approximately} represent such a state, using only $N+1$ coherent states:
\begin{align}
    |\psi\rangle \approx |\tilde{\psi}\rangle= \frac{1}{\sqrt{\mathcal{N}}} \sum_{k=0}^N c_k |\epsilon e^{2\pi i k/(N+1)}\rangle,\; c_k = \frac{e^{\epsilon^2/2}}{N+1} \sum_{n=0}^N \sqrt{n!} \frac{a_n}{\epsilon^n} e^{-2\pi i n k / (N+1)}.
    \label{eq:fock-coherent-decomp}
\end{align}
This `Fourier' decomposition is accurate to fidelity $|\langle \psi | \tilde{\psi}\rangle|^2 = 1 - O(\epsilon^{2(n+1)} / (N+1)!)$. $\mathcal{N}$ is for normalization.
In the case of a Fock basis state $|\psi\rangle = |N\rangle$, this simplifies to $|N\rangle \approx |\tilde{N}\rangle = \frac{1}{\sqrt{\mathcal{N}}} \sum_{k=0}^N e^{-2\pi i k N / (N+1)}|\epsilon e^{2\pi i k / (N+1)}\rangle$.

We will see this representation turns out to provide a more concise description in certain cases. Let us first set up some notation. We say that a state of the form $\sum_{i=1}^k c_i |\vec{\alpha}_i\rangle$ is of `coherent state rank' $k$, where $|\vec{\alpha}_i\rangle = |\alpha_{i}^{(1)}, \dots, \alpha_i^{(m)} \rangle$ is an $m$-mode coherent state. We will say an $N$ photon single-mode state has `approximate coherent rank' $N+1$, by Eq.~\eqref{eq:fock-coherent-decomp}.
It is easy to show that a linear optical (LO) unitary $U$, with transfer matrix $u$, maps a multi-mode coherent state to another multi-mode coherent state, as $U|\vec{\alpha}\rangle = |u \vec{\alpha}\rangle$, where the notation means to perform matrix-vector multiplication. This means that LO unitaries are coherent state rank preserving. 

In a LO simulation therefore, the classical complexity is governed entirely by the initial state. In the worst case (i.e., with the largest rank) for an $n$-photon simulation, the initial state has rank $k=2^n$, which is achieved for $|1\rangle^{\otimes n}|0\rangle^{\otimes(m-n)}$. This comes using Eq.~\eqref{eq:fock-coherent-decomp}, which represents $|1\rangle$ as an odd cat state, with `small' parameter $\epsilon$. The memory requirements are $O(m 2^n)$ for an $m$-mode $n$-photon simulation, and the time required to update the state under a LO unitary is $O(m^2 2^n)$. Note that this is (in general) a significant improvement over representing an arbitrary LO evolution in the Fock basis directly, which scales with the dimension $d_{n,m}=\binom{n+m-1}{n}$.

Computing transition amplitudes $\langle \vec{n}_2|U|\vec{n}_1\rangle = \overline{\langle \vec{n}_1|U|\vec{n}_2\rangle}$ costs $O(m^2 k_{min})$, where $k_{min}$ is the minimum of the approximate coherent ranks of $|\vec{n}_{1,2}\rangle$. Note that, whilst this is similar in time to permanent based methods, it is (in the worst case), exponentially worse in memory as we store the full state here. This however can have its benefits, as the coherent rank framework is a more general one, not restricted to only computing transition amplitudes under LO. For example, any operator written as a sum of products of creation and annihilation operators $\sum_i b_i \prod_{j=1}^{l_i} (a^\dag)^{n_{i,j}}(a)^{m_{i,j}}$ can be implemented, with a multiplicative overhead $O(p+1)$, where $p \ge \sum_{j=1}^{l_i} n_{i,j}, \forall i$ (that is, $p$ is the maximum number of creation operators in any single term in the sum). In particular, application of such an operator increases the states rank from $k \rightarrow k(p+1)$ (e.g.~applying $a^\dag$ will double the rank).

One use of this framework goes beyond simulation. We can show an example of this by deriving a general (exact) equation for computing LO transitions (which is equivalent to permanent based expressions \cite{aaronson_computational_2010, chabaud_quantum-inspired_2022}). Let us take two $n$ photon states $|\vec{n}_1\rangle=|n_{1,1}, \dots, n_{1,m}\rangle, |\vec{n}_2\rangle= |n_{2,1}, \dots, n_{2,m}\rangle$, with ranks $k_{1,2}$ respectively (without loss of generality, let $k_1\le k_2$). Taking the tensor product of each $|n_{1,i}\rangle$, written as a rank $n_{1,i}+1$ in the coherent basis yields
\begin{align}
    |\vec{n}_1\rangle = \lim_{\epsilon \rightarrow 0} \frac{1}{\epsilon^n} \prod_{j=1}^m \frac{\sqrt{n_{1,j}!}}{n_{1,j}+1} \sum_{i=1}^{k_1} p_i |\epsilon \vec{P}_i\rangle,
\end{align}
where $p_i$ is some phase factor $e^{i\phi_i}$, and $\vec{P_j}$ a vector of phase factors\footnote{We leave these unspecified for now as they depend on the state, but mention here they are easy to compute. For the all $|1\rangle$ state for example, $\{\vec{P}_i\}_i$ represent all bitstrings taking values $\pm 1$.}. Applying a LO unitary on $|\vec{n}_1\rangle$ is equivalent to multiplying $\vec{P}_i \rightarrow u \vec{P}_i = \vec{P}_i'$, where $u$ is the transfer matrix. Note that $\vec{P}_i'$ is no longer a vector of phases. The final overlap becomes
\begin{align}
    \langle \vec{n}_2|U|\vec{n}_1\rangle = \lim_{\epsilon \rightarrow 0} \frac{1}{\epsilon^n} \prod_{j=1}^m \left( \frac{\sqrt{n_j!}}{n_j+1} \right) \sum_{i=1}^{k_1} p_i \langle \vec{n}_2|\epsilon \vec{P}'_i\rangle = \prod_{j=1}^m \left( \frac{1}{n_{1,j}+1} \frac{\sqrt{n_{1,j}!}}{\sqrt{n_{2,j}!}} \right) \sum_{i=1}^{k_1} p_i  \prod_{l=1}^m ({P}_{i,l}')^{n_{2,l}},
    \label{eq:transition-amp}
\end{align}
where ${P}'_{j,l}$ is the $l$'th component of $\vec{P}'_j$. Notice the final term is entirely independent of $\epsilon$. This therefore provides an $O(m^2 k_1)$ algorithm form computing LO transition amplitudes (the $m^2$ scaling comes from computing each $\vec{P}_i'$ by matrix-vector multiplication). We further comment that this can be done without explicitly storing the full state as this can be computed term by term, and so the memory requirements are just $O(m^2)$ for storing the transfer matrix. That is, we can recover similar scaling in time and memory as permanent based methods, directly from calculations using the coherent rank framework.

As a last observation on this point, noting that $\langle \vec{1}|U|\vec{1}\rangle = \mathrm{perm}[u]$, one can show via Eq.~\eqref{eq:transition-amp} and results from Ref. \citenum{marshall_simulation_2023} that an alternative formula for the permanent is
\begin{align}
    \mathrm{perm}[A] = \frac{1}{2^n} \sum_{i=1}^{2^n} \prod_{j=1}^n b_{i,j} b_{i,j}'
    \label{eq:perm}
\end{align}
where $A$ is an arbitrary (not necessarily unitary) $n\times n$ complex matrix, and the sum is over all $\pm 1$ bitstrings $\{\vec{b}_i\}_i$, and $\vec{b}_i' = A \vec{b}_i$. The time complexity\footnote{Since for each binary vector $\vec{b}_i$, there is the conjugate $-\vec{b}_i$, we only need to compute $A\vec{b}_i$ for half of the total terms.} is $O(n^2 2^{n-1}$). In fact, Eq.~\eqref{eq:perm} is a known relation called Glynn's formula \cite{glynn_permanent_2010}, and has also been derived previously using the coherent state representation \cite{chabaud_quantum-inspired_2022}.

Finally we comment on performing probabilistic measurement sampling in the coherent rank framework. Whilst we saw above the cost to compute a transition probability $P_{\vec{n}}=|\langle \vec{n}|U|\vec{n}_0\rangle|^2$ is $O(m^2 k_0)$, this does not imply sampling from the distribution has the same cost (that is, outputting a Fock basis measurement sample $\vec{n}$ with probability $P_{\vec{n}}$). The standard method to do this would be by conditional probability sampling, outputting random measurement results for each mode at a time, and re-normalizing the state.
The cost for producing an $n$ photon measurement result is $O(nm \mathcal{N})$, where $\mathcal{N}$ is the cost of computing the norm of a general (unnormalized) state in the coherent basis representation $\sum_{i=1}^k c_i |\vec{\alpha}_i\rangle$. Due to the non-orthogonality of coherent states, this takes time $O(m k^2)$.  In the worst case therefore, the cost to sample from an $n$ photon LO evolved state is $O(nm^2 4^n)$, though some improvements can be found in specific cases, see Ref. \citenum{marshall_simulation_2023} for a more detailed discussion.

At this point one can make an interesting observation, for the same general argument above appears to hold for stabilizer rank simulations, in the qubit setting \cite{bravyi_simulation_2019}. In particular, stabilizer states are also non-orthogonal, and so naively computing the norm of such a `state' scales as $r^2$, where $r$ is the stabilizer rank. However, it is possible to compute such quantities in a time $O(r)$, and therefore produce probabilistic measurement results `efficiently' (i.e. in a time linear in the rank). This works by computing overlaps with random stabilizer states $|s\rangle$, the mean of which, $\mathbb{E}_{s}|\langle s | \psi\rangle|^2$, is proportional to the norm squared. Since stabilizer states form a 2-design, the convergence is `fast', and the algorithm scales linearly in $r$ (to a fixed accuracy). In Ref. \citenum{marshall_simulation_2023} it is shown that sampling with respect to random coherent states produces the correct mean, but the convergence is slow, since coherent states only form a 1-design \cite{blume-kohout_curious_2014}.

Let us however consider a more general set-up, and assume we want to estimate the norm of some (unnormalized) state of $n$ photons, $|\psi\rangle$. We can attempt to compute its norm by random sampling over unitaries $V$, from \textit{some} distribution.
In particular, we wish to compute $\mathbb{E}_V |\langle \vec{n}|V^\dag |\psi \rangle|^2$ where $|\vec{n}\rangle$ is some $n$ photon state in the Fock basis, which the specific form will turn out to be unimportant.
Indeed, if the distribution over which the unitaries are drawn form a 1-design, then one has $\mathbb{E}_V |\langle \vec{n}|V^\dag |\psi \rangle|^2 = \|\psi\|^2/d_{n,m}$,
using the 1-design result $\int d\mu_V V^\dag A V = \mathrm{Tr}(A) \mathbb{I}/d$, where $d$ is the dimension \cite{miszczak_symbolic_2017}. Notice this is (up to the dimension factor), the exact quantity of interest; the norm squared\footnote{We use notation $\|\psi\|^2 = \langle \psi|\psi\rangle$.} of `state' $|\psi\rangle$.
If the unitaries over which we were sampling also turned out to be a 2-design, one can additionally show that the variance of sampling the random variable $d |\langle \vec{n}|V^\dag|\psi\rangle|^2$ is $\sigma^2 = \frac{d-1}{d+1} \|\psi\|^4$ (here we drop the $n,m$ dependence on the dimension). This implies, in order to achieve a sample standard deviation to accuracy a fraction of the norm squared, $\sigma_T = \epsilon \| \psi \|^2$, requires $T = \frac{d-1}{\epsilon^2(d+1)}$ samples. Crucially, this quantity does not depend on the norm itself, nor scale prohibitively in the dimension. Assuming each overlap can be computed in a time linear in the states coherent rank $k$, the norm could therefore be estimated in a time that also scales linearly, in $O(k \epsilon^{-2})$ time, overcoming the non-orthogonality issue discussed above.

As we will see below, whilst LO unitaries do form a 1-design (hence the norm can be computed by sampling as above), they do not form a 2-design. This means, unfortunately for us, the fast convergence to the mean is not guaranteed. Whilst in principle one could supplement the LO unitaries with photon-number conserving non-linear operations (such as SNAP operations \cite{krastanov_universal_2015}), their implementation in the coherent rank framework is prohibitively costly (i.e., it would quickly surpass the $O(k)$ saving from just directly computing the norm). In particular, each SNAP gate increases the coherent rank by a factor of $n+1$ \cite{marshall_simulation_2023}.

\section{Unitary Designs, Universality, and Linear Optics}\label{sec:designs supersection}
In the last decade, unitary designs have emerged as a powerful tool in quantum information theory with a multitude of applications \cite{dankert2009exact,ambainis_quantum_2007,scott_optimizing_2008,roy_unitary_2009,roberts_chaos_2017,mele_introduction_2023}. Originally discovered in the context of randomized benchmarking \cite{dankert2009exact}, they have found their way into the daily toolkit of quantum information and computation. Some key applications include classical shadow tomography \cite{huang_predicting_2020}, quantum information scrambling \cite{roberts_chaos_2017}, quantum chaos \cite{leone_isospectral_2020}, decoupling methods \cite{szehr2013decoupling}, exponential speedups in query complexity \cite{brandao_exponential_2013}, optimal quantum process tomography schemes \cite{scott_optimizing_2008}, quantum analogues of universal hash functions \cite{cleve_near-linear_2016}, among others \cite{mele_introduction_2023}.

In this section, we begin with a brief overview of the theory of unitary designs (Sec.~\ref{sec:design overview}), followed by more exposition on the group $U_{n,m}$ of linear optical unitaries (Sec.~\ref{sec:LOUs as group}) and the SNAP gates (Sec.~\ref{sec:snap}). A more detailed exposition can be found in an upcoming work \cite{Anand2024_qudit_designs}. We then give several results related to unitary $t$-designs, universal gate sets, and linear optics. 
Namely, in Sec.~\ref{sec:continuous} we discuss a general result that any infinite closed subgroup of $U(d)$ for $d\geq 2$ is universal. 
This result is known in the mathematics literature \cite{katz2004larsen} without the terminology of $t$-designs, so we give an exposition in Appendix~\ref{app:continuous} for the sake of completeness. 
In Sec.~\ref{sec:design results}, we then 
show that the linear optical unitaries form a $1$-design but not a $2$-design, and that any $2$-design containing $U_{n,m}$ is in fact universal. 
(The latter being a direct corollary of the result of Sec.~\ref{sec:continuous}.) 
We then recall a result \cite{oszmaniec2017universal} identifying when extensions of $U_{n,m}$ are universal and apply it to prove that $U_{n,m}$, augmented with any single nontrivial SNAP gate, is universal. 

\subsection{Unitary designs}\label{sec:design overview}
Unitary designs mimic uniformly distributed unitaries. To determine how "uniformly distributed" an ensemble of unitaries/states is, we \textit{define} as "maximally uniform" the distribution of unitaries/states that are distributed according to the Haar measure (on the unitary group/state space, respectively). Then, a distance from this ensemble provides a quantitative metric of uniformity. Many of the technical results here are commonplace in the theory of Haar integrals and unitary designs, see e.g. Ref. \citenum{mele_introduction_2023} for a wonderful introduction to Haar measure and related tools in quantum information theory.

Recall that the Haar measure \(\mu\) is the unique group-invariant, normalized measure on \(U(d)\). That is, $\mu(U(d)) = 1$, and for all $V\in U(d)$ and $S\subseteq U(d)$, $\mu (VS) = \mu(SV) = \mu(S)$. Suppose we are given an ensemble of unitaries, \(\mathcal{E} = \{ p_{j}, U_{j} \}_{j}\), where \(p_{j} \geq 0, \sum_{j}^{} p_{j} = 1\) is a probability distribution. For our purposes, let us assume all the unitaries are equally likely, i.e., \(p_{j} = \frac{1}{\left| \mathcal{E} \right|}\), where \(\left| \mathcal{E} \right|\) is the cardinality of the ensemble. We will primarily focus on discrete ensembles, although the results can be easily generalized to continuous ones (by replacing \(p_{j}\)'s with an appropriately normalized measure).
Define the \(t\)-fold \emph{twirl} of the ensemble \(\mathcal{E}\) as the quantum channel $\Phi_{\mathcal{E}}^{(t)} (\cdot) \equiv \sum\limits_{j}^{} p_{j} U_{j}^{\dagger \otimes t} (\cdot) U^{\otimes t}.$ 
The ensemble \(\mathcal{E}\) is a unitary \(t\)-design if and only if the action of the twirling channel \(\Phi_{\mathcal{E}}^{(t)}\) is identical to that corresponding to the Haar ensemble, \(\Phi^{(t)}_{\mathrm{Haar}}\), for all linear operators. Namely,
\begin{align}
\Phi_{\mathcal{E}}^{(t)}(X) = \Phi_{\mathrm{Haar}}^{(t)}(X) ~~\forall X \in \mathcal{L}(\mathcal{H}^{\otimes t}),
\label{eq:twirling-defn-unitary-design}
\end{align}
where \(\Phi^{(t)}_{\mathrm{Haar}} (\cdot) := \int_{\mathrm{Haar}} dU U^{\dagger \otimes t} (\cdot) U^{\otimes t}\) and the integral is over the Haar measure on the unitary group. 
If one is interested in Haar averages over functions that are homogenous polynomials of degree \textit{at most} \(t\) in the matrix elements of \(U\), and at most degree \(t\) in the complex conjugates (equivalently the matrix elements of \(U^\dagger\)), then one can replace such an average with averages over a unitary \(t\)-design. 
An equivalent criterion for \(t\)-designs is in terms of Pauli operators: an ensemble \(\mathcal{E}\) is a \(t\)-design if and only if 
\eqref{eq:twirling-defn-unitary-design} holds for all $X\in\mathcal{P}^{\otimes t}$, 
where \(\mathcal{P}\) is the Pauli group. This follows from the fact that the Pauli operators form a (Hilbert-Schmidt orthogonal) basis for the operator space, and the linearity of the twirling operation.

At this point it is worth mentioning that every unitary \(t\)-design is automatically a unitary \(t'\)-design for \(t' \leq t\) with \(t,t' \in \mathbb{N}\). Therefore, in many cases, one is only interested in the `maximal' order of \(t\)-design that an ensemble of unitaries can form, since all lower orders are automatically implied.

Let us consider some examples of unitary designs. A unitary \(1\)-design is simply an ensemble of unitaries that form an orthonormal basis for \(\mathcal{L}(\mathcal{H})\). For example, Pauli matrices \(\{ I, X, Y, Z \}\) and their tensor products form a \(1\)-design for \(n\) qubits. The Pauli group on \(n\) qubits is also a \(1\)-design; we note that the additional phases \(\{ \pm 1, \pm i \}\) required for the group structure are not necessary for the \(1\)-design condition. 
Similarly, the clock-and-shift matrices for qudits (along with their tensor products) serve as a \(1\)-design for multiqudits. 

The most commonly studied example of unitary \(2\)-design is the Clifford group \(Cl_{n}\), the normalizer of the Pauli group. 
The Clifford group on \(n\) qubits forms both a \(2\)- and a \(3\)-design \cite{zhu_multiqubit_2017, webb_clifford_2016, kueng_qubit_2015}. In fact, the (qubit) Clifford group is the only infinite family of group \(3\)-designs. \cite{bannai_unitary_2018} For qudit dimension $d>2$ which is a prime power, one can analogously define the Clifford group as the normalizer of the Heisenberg-Weyl group; there, the Clifford group forms (at most) a \(2\)-design. 

The search for unitary \(t\)-designs has mostly centered around "unitary \(t\)-groups," unitary \(t\)-designs  that also form a group. The representation theory of finite groups has allowed for the complete classification of such objects, see e.g. Ref. \citenum{bannai_unitary_2018}. The key idea is the 
representation-theoretic characterization of group $t$-designs given below
(see e.g., Theorem 3 of Ref. \citenum{gross_evenly_2007}, Proposition 3 of Ref. \citenum{zhu_clifford_2016}, or Proposition 34 of Ref. \citenum{mele_introduction_2023}). 
We now give this and several standard equivalent conditions, referring to Appendix~\ref{app:rep theory} and \ref{app:rep and design exposition} for the representation-theoretic terminology and more detailed discussion. 
This result also mentions a characterization in terms of the \emph{frame potential} $F^{(t)}_{\mc{E}}$ of an ensemble $\mathcal{E}$ of unitaries, which 
is given by $F^{(t)}_\mc{E} = \frac{1}{\left| \mathcal{E} \right|^{2}} \sum\limits_{U,V \in \mathcal{E}}^{} \left| \operatorname{Tr}\left[ U^{\dagger}V \right] \right|^{2t}$ in the discrete case and $F^{(t)}_\mc{E} = \frac{1}{\mu(\mc{E})^2} \int_{\mc{E}}\int_{\mc{E}} d\mu(U)d\mu(V)\left| \operatorname{Tr}\left[ U^{\dagger}V \right] \right|^{2t}$ in the continuous case (where $\mu$ is the Haar measure on $U(d)$). 
(See Ref.~\citenum{mele_introduction_2023} for a general treatment.) 
For $U(d)$ and $t\leq d$, we have
$F^{(t)}_{U(d)} = t!$ \cite{zhu_clifford_2016}.

\begin{proposition}
\label{prop:designs-equivalent-defns}
Let \(G \subseteq U(d)\) be a (finite or compact) subgroup of the unitary group. The following are equivalent:
\begin{enumerate}
    \item \(G\) is a unitary \(t\)-design.
    \item The \(t\)-copy diagonal action of \(G\), denoted as \(\tau^{t}(G)\), decomposes into the same number of irreps as the \(t\)-copy diagonal action of \(U(d)\). 
    \item \(\tau^{t}(G)\) has the same commutant as \(\tau^{t}(U(d))\).
    \item $F^{(t)}_{G} = F^{(t)}_{U(d)}$. 
    \label{cond:frame potential}
\end{enumerate}
\end{proposition}
The ability to check if a subgroup of the unitary group forms a \(t\)-design then reduces to understanding its representation theory. 
In particular, writing $\mathcal{V} = \mathbb{C}^d$, $G$ is a $1$-design if and only if $\mathcal{V}$ is an irreducible $G$-representation, and $G$ is a $2$-design if and only if $\mathcal{V}\otimes \mathcal{V}$ has exactly $2$ irreducible $G$-subrepresentations 
(corresponding to the symmetric and anti-symmetric subspaces).

\subsection{The group of linear optical unitaries}\label{sec:LOUs as group}

As in Sec.~\ref{sec:loss dist}, we consider the Fock space $\mathcal{H} = \mc{H}_{n,m}$ of $n$ photons in $m$ modes, with dimension $d_{n,m} = \binom{n+m-1}{n}$ 
and Fock basis $\{ \ket{n_0, \dots, n_{m-1}}: \sum_i n_i = n, n_i\geq 0\}$. 
We may view $\mathcal{H}$ as the \emph{symmetric subspace} of $\left(\mathbb{C}^m\right)^{\otimes n}$, as follows. 
First, each photon has an $m$-dimensional state space $\mathbb{C}^m$ describing its position among the $m$ modes. 
The state space of $n$ such photons is contained in the appropriate tensor product, $\left(\mathbb{C}^m\right)^{\otimes n}$, 
with each copy of $\mathbb{C}^m$ indexing a different photon. 
Since photons are bosons, given a state of $n$ (ideal, indistinguishable) photons, one should not be able to determine which photon is which. 
Thus, only vectors that are symmetric with respect to permutations of the tensor factors are legitimate photonic states. 
We sometimes find it convenient, then, to view 
\begin{align}\label{eq:h symmetric}
    \mathcal{H} = \left( \left(\mathbb{C}^m\right)^{\otimes n}  \right)^{S_n}.
\end{align}
We also refer the reader to Ref. \citenum{harrow_church_2013} for an excellent review on the symmetric subspace and its applications in quantum information theory. 

\begin{remark}\label{rem:symm basis}
Let $v_0, \dots, v_{m-1}$ be a basis for the single-photon state space $\mathbb{C}^m$. 
From this perspective, the Fock basis vector $\ket{n_0, \dots, n_{m-1}}$ corresponds to the (normalized, unit vector) symmetrization of  $v_{m_1}\otimes\cdots\otimes v_{m_n},$  where $m_1, \dots, m_n$ is a weakly increasing list of integers in $[0,m-1]$ containing $n_0$ copies of $0$, $n_1$ copies of $1$, etc. 
In other words, whereas the Fock basis expression $\ket{n_0, \dots, n_{m-1}}$ tracks the number of photons in each mode, the above expression tracks, for each photon, which mode it is in (then symmetrizes so that the photons are indistinguishable). 
For example, in the case $n=m=2$, we have the correspondence
\begin{align*}
    \ket{2,0} = v_0\otimes v_0,\;
     \ket{0,2} = v_1\otimes v_1,\;
     \ket{1,1} = \frac{1}{\sqrt{2}}\left( v_0\otimes v_1 + v_1\otimes v_0  \right).
\end{align*}
\end{remark}
From this perspective, we may give a natural characterization of the group of linear optical unitaries as a subgroup of $U(\mathcal{H})$, the full unitary group of $\mathcal{H}$. 
We consider the unitary group $U(m)$, which naturally acts on the state space $\mathbb{C}^m$ of an individual photon. 
For $U\in U(m)$, $U^{\otimes n}$ acts on $\left(\mathbb{C}^m\right)^{\otimes n}$ and commutes with permutations of the tensor factors. Then by \eqref{eq:h symmetric}, we obtain an injective group homomorphism
\begin{align}\label{eq:LOU map}
    U(m)\ni U\mapsto U^{\otimes n}\in U(\mathcal{H}).
\end{align}
We let $U_{n,m}\cong U(m)$ denote the image of this homomorphism, the group of linear optical unitaries. 
Generally speaking, the results of this section are concerned with how $U_{n,m}$ sits inside $U(\mathcal{H})$. 
The cases $n=1$ and $m=1$ are trivial, as $U_{1,m} = U(m) = U(\mathcal{H})$ and $U_{n,1} = U(1)$. 
As these two cases exhibit rather different behavior from the general case, going forward we assume $n,m \geq 2$. 

\subsection{SNAP gates}\label{sec:snap}

We now briefly introduce another set of elements of $U(\mathcal{H})$, the SNAP (selective number-dependent arbitrary phase) gates. 
As we will see in Theorem~\ref{thm:zimboras}, linear optics assisted by nearly any other unitary is sufficient to generate a universal gate set in $U(\mathcal{H})$. 
In particular, we will show in Theorem~\ref{thm:snap} that linear optics with any nontrivial SNAP gate is sufficient for universality. 
We briefly remark on our choice of the SNAP gate as the "resourceful" gate on top of linear optical unitaries. For the single mode case this shows up naturally in cavity-QED systems where a cavity dispersively coupled to a qubit is used to perform quantum computation \cite{krastanov_universal_2015}. 
It was shown here that displacements on the infinite-dimensional Hilbert space of a single mode bosonic system along with SNAP gates is sufficient to perform universal quantum computation on any \(d\)-dimensional subspace of the bosonic mode. Namely, one can realize universal control of a single qudit via this approach. Our use of the SNAP gates is inspired by this result. 

For $0\leq k\leq n$, $0\leq s\leq m-1$, and $\theta\in\mathbb{R}$, define the SNAP gate
\begin{align}\label{def:snap}
    S_{k,s}(\theta) = \exp(i\theta \ketbra{k}_s),
\end{align}
where the subscript indicates that the action is on the $s$th mode. 
Note that the SNAP gates are diagonal in the Fock basis. 
Of course, if $\theta\equiv 0\mod 2\pi$, then $S_{k,s}(0)$ is the identity operator. 
Note that if $\theta \neq 0\mod 2\pi$, then $S_{k,s}(\theta)\not\in U_{n,m}$. 
This is quite intuitive using the symmetric subspace perspective above. 
For example, recalling the case $m=n=2$ from Remark~\ref{rem:symm basis}, we see that $S_{2,1}(\theta)$ has the following action: 
\begin{align*}
    v_0\otimes v_0 = \ket{2,0}&\mapsto e^{i\theta}\ket{2,0} = e^{i\theta} v_0\otimes v_0
    \\ v_1\otimes v_1 = \ket{0,2} &\mapsto \ket{0,2} = v_1\otimes v_1
    \\ \frac{1}{\sqrt{2}}\left( v_0\otimes v_1 + v_1\otimes v_0  \right) = \ket{1,1} &\mapsto \ket{1,1} =  \frac{1}{\sqrt{2}}\left( v_0\otimes v_1 + v_1\otimes v_0  \right).
\end{align*}
This clearly cannot arise from an action of the form $U\otimes U$. 

\subsection{Result: Continuous 2-designs}\label{sec:continuous}

In this section, we observe that there are no non-universal infinite closed (i.e. continuous) $2$-designs in $U(\mathcal{V})$ if $\dim\mathcal{V}\geq 2$. 
This is essentially Theorem 1.1.6(2) in Ref. \citenum{katz2004larsen}, and we closely follow the proof given there. 
This result was not stated or proved in the context of unitary $t$-design theory, 
so for the sake of accessibility we choose to restate it here and prove it in Appendix~\ref{app:continuous}. 
The main idea of the proof is as follows. First, due to Proposition~\ref{prop:designs-equivalent-defns}, $G\subseteq U(\mathcal{V})$ is a $2$-design if and only if $\mathcal{V}\otimes \mathcal{V}$ decomposes into exactly $2$ irreducible representations for $G$. 
Second, following Ref. \citenum{katz2004larsen}, one can prove that $\mathcal{V}\otimes \mathcal{V}^*$ decomposes into exactly $2$ irreducible representations for $G$ if and only if $SU(d)\subseteq G$. 
(Note that $\mathcal{V}\otimes \mathcal{V}^*\cong \mathfrak{gl}(\mathcal{V})$, which is very useful in the proof of this result.) 
To obtain the result, then, one needs only to show that $\mathcal{V}\otimes \mathcal{V}$ and $\mathcal{V}\otimes \mathcal{V}^*$ decompose into the same number of irreducible $G$-representations. 
This is a much more general fact that holds whenever the finite-dimensional representations of $G$ are known to be completely reducible. 

\begin{theorem}\label{thm:no 2 designs}
    Let $\mathcal{V}\cong\mathbb{C}^d$, with $d\geq 2$. 
    Let $G$ be an infinite closed subgroup of $U(d)=U(\mathcal{V})$. If $G$ is a unitary $2$-design, then $SU(d)\subseteq G$. In particular, any set of unitaries generating $G$ is a universal set. 
\end{theorem}

In the following section, we apply this result to observe that there are no non-universal extensions of linear optics to a $2$-design. 

\subsection{Results on linear optical t-designs and universality}\label{sec:design results}

Since $U_{n,m}$ is not universal in $U(\mathcal{H})$, we instead ask whether it is a $t$-design for sufficiently large $t$. 
We first observe that the linear optical unitaries form a $1$-design. This result is quite straightforward and is surely known; however, as we could not find an explicit statement in the literature, we provide the proof in Appendix~\ref{app:1 design}. 

\begin{proposition}\label{prop:1 design}
    Let $n,m\geq 2$. 
    The group $U_{n,m}$ of linear optical unitaries is a $1$-design in $U(\mathcal{H})$. 
\end{proposition}

However, $U_{n,m}$ is not a $t$-design for $t\geq 2$. This may be proven directly using representation theory, but it is also a trivial consequence of Theorem~\ref{thm:no 2 designs}:

\begin{corollary}\label{cor:2 design}
Let $n,m\geq 2$. Let $G$ be a closed subgroup of $U(\mathcal{H})$ extending the linear optical unitaries, $U_{n,m}\subseteq G\subseteq U(\mathcal{H})$. 
If $G$ is a $2$-design in $U(\mathcal{H})$, then $G$ is universal. 
In particular, since $U_{n,m}$ is not universal, it is not a $t$-design for $t\geq 2$. 
\end{corollary}

As a result of Corollary~\ref{cor:2 design}, if we want a $2$-design extending $U_{n,m}$, we are forced to work with a universal set. 
The natural question is whether we can at least obtain a nice generating set for $U(\mathcal{H})$. 
Oszmaniec and Zimbor\'{a}s \cite{oszmaniec2017universal} showed that linear optics augmented by nearly any additional gate is sufficient to densely generate $U(\mathcal{H})$. We further discuss this result in Appendix~\ref{app:zimboras}. 

\begin{theorem}[\cite{oszmaniec2017universal}]
\label{thm:zimboras}
Consider the group $U_{n,m}$ of linear optical unitaries for $n\geq 2$ photons in $m\geq 2$ modes. Let $V\in U(\mathcal{H})$ be any unitary gate with $V\not\in U_{n,m}$. 
   \begin{enumerate}
        \item For $m>2$, the group generated by $U_{n,m}$ and $V$ is universal. 
        \item For $m=2$, define 
            $\ket{\zeta} = \sum_{a=0}^n (-1)^a \ket{a, n-a}\otimes\ket{n-a, a}.$
        If 
        \begin{align}\label{eq:zeta comm}
            [V\otimes V, \ketbra{\zeta}] \neq 0,
        \end{align}
        then the group generated by $U_{n,m}$ and $V$ is universal. 
    \end{enumerate}
\end{theorem}

In Appendix~\ref{app:snap}, we prove that the extension of linear optics by any single SNAP gate at a fixed (nontrivial) angle, on a fixed mode $s$, and acting upon a fixed photon number $k$, is sufficient to give universality: 

\begin{theorem}\label{thm:snap}
Let $n,m\geq 2$, $0\leq k\leq n$, $0\leq s\leq m-1$. 
Let $\theta\in\mathbb{R}$ with $\theta\not\equiv 0\mod 2\pi$. 
Then the group generated by $U_{n,m}$ and $S_{k,s}(\theta)$ is universal. 
\end{theorem}

In Corollary 1 of Ref. \citenum{boulandGenerationUniversalLinear2014}, Bouland proved that if $m\geq 3$, one may densely generate $U_{n,m}$ using copies of a single beamsplitter. 
Combined with 
Theorem~\ref{thm:snap}, this implies that $U(\mathcal{H})$ may be generated using a single SNAP gate (with fixed mode, angle, and photon number) and a single beamsplitter (acting on different pairs of modes).

\section{Discussion}\label{sec:discussion}
In Sec.~\ref{sec:loss dist}, we discussed paradigms for the simulation of loss and distinguishability errors in photonics, and we applied these paradigms to investigate the impact of distinguishability errors on Type I fusion, generalized Type II fusion (the $n$-GHZ state analyzer), and $n$-GHZ state generation. 
For low-order errors, one may do the calculations by hand, as in Secs. \ref{sec:fusion_dist_statements} and \ref{sec:dist ghz}. 
To incorporate higher-order errors in Bell state generation in Sec.~\ref{sec:numerics}, we turned to numerical simulation. 
In the longer term, it will be essential to simulate circuits larger than the BSG circuit, an obvious example being $n$-GHZ state generation. More generally, as in fusion-based quantum computation, one often wants to generate complex linear optical states by starting from small entangled states, such as GHZ states, and applying entangling operations, such as fusion \cite{browne_resource-efficient_2005, bartolucci2021creation, bartolucci2023fusion}. 
In order to understand the impact of distinguishability and other errors on these more complex circuits, a variety of techniques may be employed. 
Of course, for sufficiently small circuits we may proceed directly as in Sec.~\ref{sec:numerics}. 
Our methods may be applied to the $3$-GHZ case, for example, with little change. 
In certain circumstances, we may achieve more efficient simulations by using the coherent state approximations of Sec.~\ref{sec:coherent}. 
This can be used to expand the range of simulable circuits. 
For sufficiently large circuits, however, none of these paradigms will be sufficient. 
Instead, we may use the low-order approximations of Theorem~\ref{thm:ghz gen}, and especially the ideas of Remark~\ref{rem:stabilizer sims}, to develop larger-scale low-order simulations. 
When the error rates are sufficiently small, these techniques would allow for a stabilizer-type simulation, where generated GHZ states are viewed as ideal states with sampled Pauli and distinguishability errors applied. 
Such calculations are especially clean when suitably simple distinguishability error models such as the OBB, SBB, and OBP models of Sec.~\ref{sec:numerics} are used.

In Sec.~\ref{sec:coherent} it was argued that the fact that linear optical (LO) unitaries only form a 1-design (recall Prop.~\ref{prop:1 design} and Cor.~\ref{cor:2 design}) has implications on the classical simulability of Fock basis measurements in LO. It was also shown in Th.~\ref{thm:snap} that including SNAP operations promotes the group to universality (or in other words, the augmented group forms a $t$-design for all $t$). However, this statement is a slight idealization, for in a practical setting one cannot apply an arbitrary number of SNAP gates. In fact, SNAP gates are quite expensive when implemented via LO and post-selection (required to enact effective non-linearity); for example, the SNAP gate $S_{n,s}(\pi)$, known as the non-linear sign (NS) gate \cite{kok_linear_2007}, is implemented only with probability $1/n^2$. Clearly one cannot expect to apply too many of these. One pertinent question is the number of SNAP (or NS) gates required to promote LO to an approximate $2$-design. (We refer the reader to Refs.  \cite{zhu_clifford_2016, mele_introduction_2023} for definitions of approximate $t$-designs and related notions.) 
In Fig.~\ref{fig:welch_ratio_LO}, we numerically study this via the 
\emph{Welch ratio}, the ratio of frame potentials 
$F^{(2)}_{\mc{E}}/F^{(2)}_{U(d)}$, 
where $\mc{E}$ is an ensemble of products of LO unitaries and $N_{NS}$ NS gates. 
This ratio must approach $1$, since $U_{n,m}$ with NS gates is universal and thus a $2$-design; the question is the rate of convergence in terms of $N_{NS}$.

\begin{figure}
    \centering
    \includegraphics[width=0.45\columnwidth]{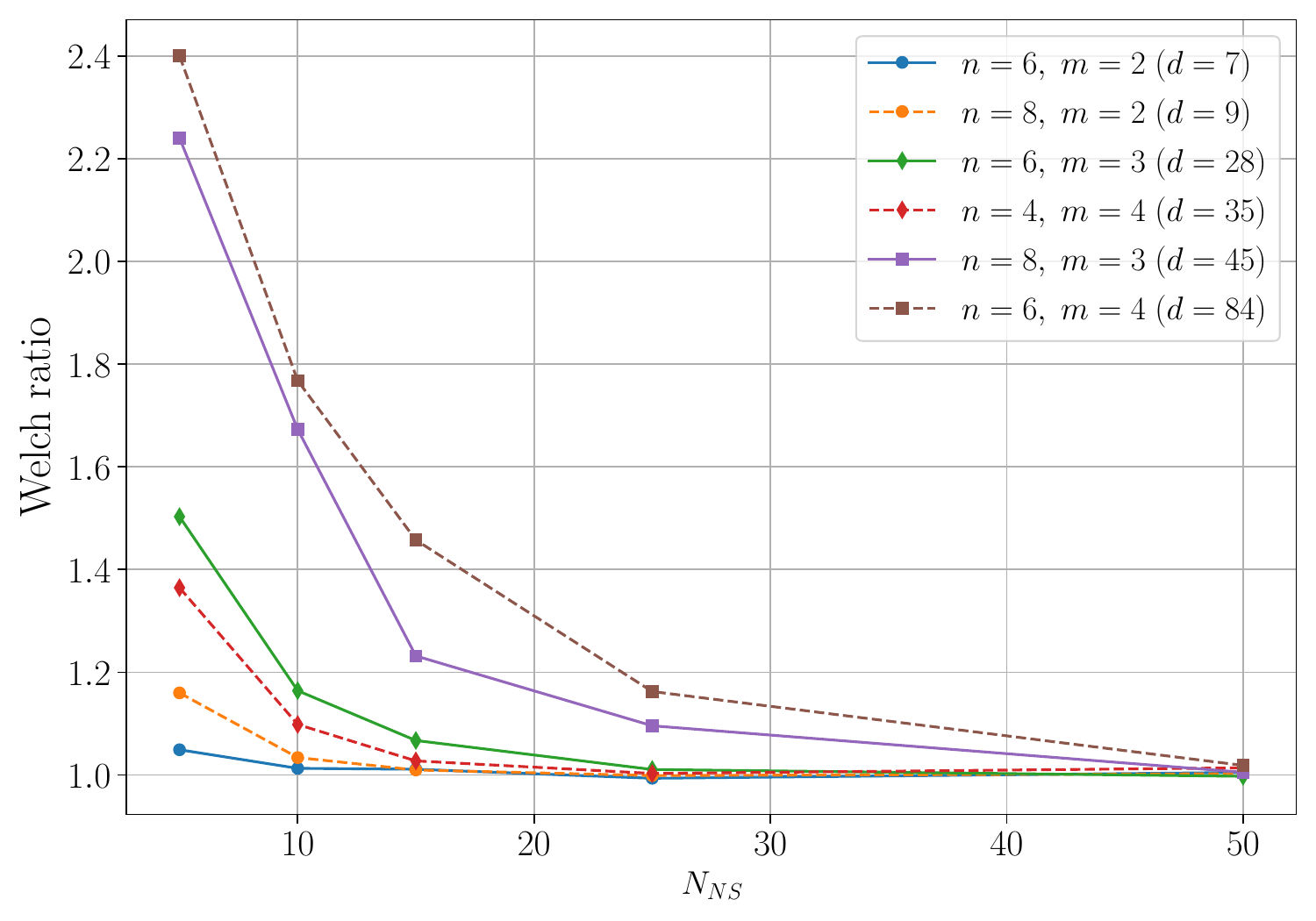}
    \caption{We plot the Welch ratio for ensembles of unitaries of the form $U_1 S_{n,s_1}(\pi)\cdots U_{N} S_{n,s_{N}}(\pi) U_{N+1}$; 
    the $U_i\in U_{n,m}$ are Haar-random, the modes $s_i$ of the NS gates are chosen uniformly at random, and $N=N_{NS}$ counts the number of NS gates. 
    We plot the ratio 
    for varying values of $n$ and $m$ shown in the legend, with $d=d_{n,m}$. 
    }
    \label{fig:welch_ratio_LO}
\end{figure}

\acknowledgments 
 
We are grateful for support from the NASA SCaN program, from DARPA under IAA 8839, Annex 130, and from  NASA Ames Research Center.
J.M. is thankful for support from NASA Academic Mission Services, Contract No. NNA16BD14C. N.A. is a KBR employee working under the Prime Contract No. 80ARC020D0010 with the NASA Ames Research Center. The United States Government retains, and by accepting the article for publication, the publisher acknowledges that the United States Government retains, a nonexclusive, paid-up, irrevocable, worldwide license to publish or reproduce the published form of this work, or allow others to do so, for United States Government purposes.

\bibliography{report} 

\appendix
\section{Resource state generation}\label{sec:rsg proofs}

\subsection{Proof sketches for Sec.~\ref{sec:fusion_dist_statements}}

We briefly sketch the proofs of the results in Sec.~\ref{sec:fusion_dist_statements}. 

Lemma~\ref{lemma:x measurement} is immediate from the fact that success is heralded if and only if the measurement patterns $(1,0), (0,1)$ are obtained. 

For Proposition~\ref{prop:x measurement}, we begin with 
the terms $c_{10}\ket{\phi_{10}}\cre{0}(\xi)\vac + c_{01}\ket{\phi_{01}}\cre{1}(\xi')\vac$ and apply the beamsplitter to obtain
\begin{align}
    ((c_{10}\ket{\phi_{10}}\cre{0}(\xi) + c_{01}\ket{\phi_{01}}\cre{0}(\xi')) + (c_{10}\ket{\phi_{10}}\cre{1}(\xi) - c_{01}\ket{\phi_{01}}\cre{1}(\xi')))\vac.
\end{align}
The measurement pattern $(1,0)$ projects onto the first term, while $(0,1)$ projects onto the second. 
We consider the first case; the second will be similar. 
Expressing the resulting state in terms of the occupation number and internal states simultaneously (referred to as the logical-internal basis in Ref. \citenum{sparrowQuantumInterferenceUniversal2017}), we obtain
\begin{align}
    \left(c_{10}\ket{\phi_{10}}\ket{1,0}\ket{\xi} + c_{01}\ket{\phi_{01}}\ket{1,0}\ket{\xi'}\right).
\end{align}
Taking the corresponding mixed state, then taking the partial trace over the internal states in the two input modes (and dropping the now-redundant $\ket{1,0}$ factor), we obtain
\begin{align*}
    &|c_{10}|^2\braket{\xi}{\xi}\ketbra{\phi_{10}} 
    + |c_{01}|^2\braket{\xi'}{\xi'}\ketbra{\phi_{01}}
    + c_{10} \overline{c_{01}}\braket{\xi'}{\xi}\ketbra{\phi_{10}}{\phi_{01}}  
    + c_{01} \overline{c_{10}}\braket{\xi}{\xi'}\ketbra{\phi_{01}}{\phi_{10}}  
    \\=& |c_{10}|^2\ketbra{\phi_{10}} 
    + |c_{01}|^2\ketbra{\phi_{01}}
    + \frac{1}{2}\textnormal{Re}(c_{10} \overline{c_{01}}\braket{\xi'}{\xi}\ketbra{\phi_{10}}{\phi_{01}}),
\end{align*}
giving \eqref{eq:dist x measurement general} as desired. 
The rest of the proposition follows from considering the cases $\braket{\xi'}{\xi}=1,0$. 

Proposition~\ref{prop:type i} is a direct application of Proposition~\ref{prop:x measurement} to the Type I fusion circuit in Figure~\ref{fig:fusion}(b). 

As with Lemma~\ref{lemma:x measurement}, Lemma~\ref{lemma:ghz filtering} is a direct consequence of the success conditions for the $n$-GHZ state analyzer, detailed in Figure~\ref{fig:ghz_analyzer}. 

For Theorem~\ref{thm:ghz analyzer}, we simply note that the $n$-GHZ state analyzer is a concatenation of $n$ generalized $X$ measurements. Iteratively applying Proposition~\ref{prop:x measurement} and carefully tracking the signs gives the result. 

\subsection{GHZ generation}

Here we will discuss the details of the $n$-GHZ state generation protocol. We begin with the ideal case, sketching the argument for the sake of exposition and later comparison. 
We then sketch how this argument changes for the two types of distinguishability errors studied in Theorem~\ref{thm:ghz gen}. 

\subsubsection{Proof sketch for Lemma~\ref{lemma:ghz gen}}
We now briefly outline the proof of Lemma~\ref{lemma:ghz gen}, showing how Protocol~\ref{protocol:ghz} works in the ideal case. Step~\ref{ghz:noon} performs the transformation $|1,1\rangle\mapsto \frac{1}{\sqrt{2}}\left(\ket{2,0}-\ket{0,2}\right)$
on each pair of modes $(2i, 2i+1)$. 
Then Step~\ref{ghz:splitter} distributes this state across the four modes $(2i, 2i+1, 2n+2i, 2n+2i+1)$, yielding the state
\begin{align}\label{ghz:bad bell}
    \ket{\beta} = \frac{1}{2}\left( \ket{1,0,1,0} - \ket{0,1,0,1}\right) + \frac{1}{2\sqrt{2}}\left(\ket{2,0,0,0}-\ket{0,2,0,0} + \ket{0,0,2,0} - \ket{0,0,0,2}\right),
\end{align}
which we recognize as a linear combination of a dual-rail Bell pair $\frac{1}{\sqrt{2}}\left(\ket{\mathbf{00}} - \ket{\mathbf{11}}\right)$ and ``junk'' terms that we hope to discard. 
After Step~\ref{ghz:splitter}, then, we have prepared $n$ copies of $\ket{\beta}$ from \eqref{ghz:bad bell}. 
We then feed the second half of each $\ket{\beta}$ into the $n$-GHZ state analyzer and post-select for success. 
Applying Lemma~\ref{lemma:ghz filtering}, we see that the multi-photon terms are dropped, and the operation may be viewed as taking the second half of each state in 
\begin{align}
    \frac{1}{2^{n/2}}\left( \ket{1,0,1,0} - \ket{0,1,0,1}\right)^{\otimes n},
\end{align}
and projecting onto a known $n$-GHZ state.  
It is easy to see that this results in the output state being a $n$-GHZ state, up to Pauli corrections from the projection, which are accounted for in Step~\ref{ghz:correct}. 

\subsubsection{Proof sketch for Theorem~\ref{thm:ghz gen}}

We now consider the first-order error case studied in Theorem~\ref{thm:ghz gen} above, in which 
there is a distinguishability error in the $i$th pair, modes $(2i, 2i+1)$. 
Without loss of generality, we take $i=0$, considering modes $(0,1)$. 
We have internal states $\eta_0=\ketbra{\psi_0}, \eta_1=\ketbra{\psi_1}$, where $\eta_0$ is the ideal state and $Tr(\eta_0 \eta_1) = 0$. 
We have as input to modes $(0,1)$ the state $\rho$ in which the non-ideal photon may be in either mode: thus, as in \eqref{eq:first order input}, $\rho$ is an even mixture of the two states $\cre{0}(\eta_0)\cre{1}(\eta_1)\vac$ and $\cre{0}(\eta_1)\cre{1}(\eta_0)\vac$. 
Calculating the image of each under Step \ref{ghz:noon} and simplifying, we obtain an even mixture of the states
\begin{align}
    (\creta{0}{0}\creta{0}{1} - \creta{1}{0}\creta{1}{1})\vac\textnormal{ and }(\creta{0}{0}\creta{1}{1} - \creta{0}{1}\creta{1}{0})\vac.
\end{align}
We then apply Step~\ref{ghz:splitter}, giving a state on modes $0,1,2n,2n+1$. 
Taking the dual-rail projection as discussed above Theorem~\ref{thm:ghz gen}, 
we drop all terms with $0$ or $2$ photons in modes $(0,1)$, giving a state $\rho'$ that is an even mixture of the two states
\begin{align}
    &(\creta{0}{0}\creta{2n}{1} - \creta{1}{0}\creta{2n+1}{1})\vac + (\creta{0}{1}\creta{2n}{0} - \creta{1}{1}\creta{2n+1}{0})\vac,\label{eq:dist mix 0}
    \\& (\creta{0}{0}\creta{2n+1}{1} - \creta{1}{0}\creta{2n}{1})\vac + (\creta{1}{1}\creta{2n}{0} - \creta{0}{1}\creta{2n+1}{0})\vac.\label{eq:dist mix 1}
\end{align}
(Note the dropped terms occur with equal weight in both cases, corresponding to terms in which an even number of photons travel from modes $(0,1)$ to modes $(2n,2n+1)$ via the 50:50 beamsplitters.) 
We consider \eqref{eq:dist mix 0} in detail. 
This state has been expressed as a sum of two terms: in one, the photon with fully distinguishable internal state $\eta_1$ occupies modes $(2n, 2n+1)$; in the other, it is the photon with ideal internal state $\eta_0$. 
Note that the later steps in Protocol~\ref{protocol:ghz} do not involve interactions between the mode pairs $(0,1)$ and $(2n, 2n+1)$. 
Thus, since $\Tr(\eta_0 \eta_1) = 0$, there can be no further interaction between the two terms of \eqref{eq:dist mix 0}. Reasoning similarly for \eqref{eq:dist mix 1}, we observe that it is equivalent to replace $\rho'$ with the state $\rho''$, an even mixture of the four terms 
\begin{align}
    (\creta{0}{0}\creta{2n}{1} - \creta{1}{0}\creta{2n+1}{1})\vac &= (\ket{1,0,1,0} - \ket{0,1,0,1})\otimes\eta_0\otimes\eta_1,\label{dist_ghz:0}
    \\(\creta{0}{1}\creta{2n}{0} - \creta{1}{1}\creta{2n+1}{0})\vac &= (\ket{1,0,1,0} - \ket{0,1,0,1})\otimes\eta_1\otimes\eta_0,\label{dist_ghz:1}
    \\ (\creta{0}{0}\creta{2n+1}{1} - \creta{1}{0}\creta{2n}{1})\vac &= (\ket{1,0,0,1} - \ket{0,1,1,0})\otimes\eta_0\otimes\eta_1, \label{dist_ghz:2}
    \\(\creta{1}{1}\creta{2n}{0} - \creta{0}{1}\creta{2n+1}{0})\vac &= -(\ket{1,0,0,1} - \ket{0,1,1,0})\otimes\eta_1\otimes\eta_0.\label{dist_ghz:3}
\end{align}
Here the right-hand size of each line expresses the state in terms of the occupation numbers and the internal states on modes $(0,1,2n,2n+1)$; the second factor corresponds to the internal state in modes $(0,1)$ and the third factor to the internal state in modes $(2n,2n+1)$. 
To understand the state generated by Protocol~\ref{protocol:ghz} in this case, it suffices to consider each of the terms above and apply the $n$-GHZ state analyzer. 
Each of these are in a form that is easily studied using Theorem~\ref{thm:ghz analyzer}. 
Note we are able to rule out patterns of the form $(m_{2j}, m_{2j+1}) \in\{(1,1),(0,0)\}$ for $j\neq i$, because we assume the only errors are in pair $i$. Then we can express the state in the form \eqref{eq:reduced form} and apply the theorem. 

For \eqref{dist_ghz:0} and \eqref{dist_ghz:2}, the input photon to the GHZ state analyzer in modes $(2n,2n+1)$ is fully distinguishable from the ideal. 
This has two consequences: first, the output state consists of only ideal photons, since the distinguishable one has been measured away. 
Second, by Theorem~\ref{thm:ghz analyzer} Part 4, the reported value of the measurement of the observable $X^{\otimes n}$ is fully randomized. 
Then the projection (after correction) is to the state $\ketbra{10}^{\otimes n} + \ketbra{01}^{\otimes n} = P_{Z_0}(\ketbra{B_+})$ in the case \eqref{dist_ghz:1}. 
In the case \eqref{dist_ghz:3}, the correlation between the photons in modes $(0,1)$ and $(2n,2n+1)$ is the opposite of the intended correlation, so we obtain
$X_0(P_{Z_0}(\ketbra{B_+}))X_0^\dagger$. 
Taken together, these two terms correspond to the mixed state $P_{X_0}(P_{Z_0}(\ketbra{B_+}))$. 

For \eqref{dist_ghz:1} and \eqref{dist_ghz:3}, the input photons to the GHZ state analyzer are all ideal, so the measurement is ideal, as in Theorem~\ref{thm:ghz analyzer} Part 3. 
As above, the two terms have opposite correlation, so the output state will be scrambled by $P_{X_0}$,
randomizing the output mode of the photon in modes $(0,1)$. 
Further, in both cases the photon in output modes $(0,1)$ has internal state $\eta_1$, fully distinguishable from the ideal. 
The terms \eqref{dist_ghz:1}, \eqref{dist_ghz:3} then correspond to $P_{X_0}(D_0(\ketbra{B_+}))$, where $D_i$ is as defined before Theorem~\ref{thm:ghz gen}.

In summary, replacing modes $(0,1,2n,2n+1)$ with $(2i, 2i+1, 2n+2i, 2n+2i+1)$, we 
may describe the output as 
\begin{align}
    \frac{1}{2}P_{X_i}\left( P_{Z_i}(\ketbra{B_+}) + D_i(\ketbra{B_+})  \right).
\end{align}

The pair error case is similar and simpler: after Step~\ref{ghz:splitter}, the $i$th pair is in a state like \eqref{ghz:bad bell} above, but with internal state $\eta_1$ for both photons. This then leads to scrambling of the $n$-GHZ analyzer's $X^{\otimes n}$ measurement but \emph{not} the $Z\otimes Z$ measurements; further, the output photon in modes $(2i,2i+1)$ is necessarily distinguishable, in state $\eta_1$. 

\subsection{Numerical calculations}\label{app:numerics}
Here we give the expected values of the random variables studied in Sec.~\ref{sec:numerics} above. 
We use the abbreviations $F$ for the standard fidelity, $PF$ for the post-selected fidelity, $FF$ for the post-fusion fidelity, and $NS$ for the number of stabilizer errors. The subscripts indicate the error model. 
We recall that in the case of the OBP (orthogonal bad pairs) error model, we take the probability of a pair error to be $1-(1-\epsilon)^2 = \epsilon(2-\epsilon)$, so that the probability of an ideal pair is the same as the other models, $(1-\epsilon)^2$. 

\allowdisplaybreaks

{\footnotesize

\begin{align}
    E(F_{OBB}) &= \frac{0.0625 \epsilon^{2} \left(1 - \epsilon\right)^{2} + 0.0625 \epsilon \left(1 - \epsilon\right)^{3} + 0.125 \left(1 - \epsilon\right)^{4}}{4.5 \epsilon^{4} + 4.5 \epsilon^{3} \cdot \left(1 - \epsilon\right) + 2.125 \epsilon^{2} \left(1 - \epsilon\right)^{2} + 0.625 \epsilon \left(1 - \epsilon\right)^{3} + 0.125 \left(1 - \epsilon\right)^{4}}
    \\ E(F_{SBB}) &= \frac{0.0625 \epsilon^{2} \left(1 - \epsilon\right)^{2} + 0.0625 \epsilon \left(1 - \epsilon\right)^{3} + 0.125 \left(1 - \epsilon\right)^{4}}{0.125 \epsilon^{4} + 0.625 \epsilon^{3} \cdot \left(1 - \epsilon\right) + 1.0 \epsilon^{2} \left(1 - \epsilon\right)^{2} + 0.625 \epsilon \left(1 - \epsilon\right)^{3} + 0.125 \left(1 - \epsilon\right)^{4}}
    \\ E(F_{OBP}) &= \frac{0.125 \left(- \epsilon \left(2 - \epsilon\right) + 1\right)^{2}}{0.25 \epsilon^{2} \left(2 - \epsilon\right)^{2} + 0.25 \epsilon \left(2 - \epsilon\right) \left(- \epsilon \left(2 - \epsilon\right) + 1\right) + 0.125 \left(- \epsilon \left(2 - \epsilon\right) + 1\right)^{2}}
    \\ E(PF_{OBB}) &= \frac{0.0625 \epsilon^{2} \left(1 - \epsilon\right)^{2} + 0.0625 \epsilon \left(1 - \epsilon\right)^{3} + 0.125 \left(1 - \epsilon\right)^{4}}{3.0 \epsilon^{4} + 3.0 \epsilon^{3} \cdot \left(1 - \epsilon\right) + 1.5 \epsilon^{2} \left(1 - \epsilon\right)^{2} + 0.5 \epsilon \left(1 - \epsilon\right)^{3} + 0.125 \left(1 - \epsilon\right)^{4}}
    \\ E(PF_{SBB}) &= \frac{0.0625 \epsilon^{2} \left(1 - \epsilon\right)^{2} + 0.0625 \epsilon \left(1 - \epsilon\right)^{3} + 0.125 \left(1 - \epsilon\right)^{4}}{0.125 \epsilon^{4} + 0.5 \epsilon^{3} \cdot \left(1 - \epsilon\right) + 0.75 \epsilon^{2} \left(1 - \epsilon\right)^{2} + 0.5 \epsilon \left(1 - \epsilon\right)^{3} + 0.125 \left(1 - \epsilon\right)^{4}}
    \\ E(PF_{OBP}) &= \frac{0.125 \left(- \epsilon \left(2 - \epsilon\right) + 1\right)^{2}}{0.25 \epsilon^{2} \left(2 - \epsilon\right)^{2} + 0.25 \epsilon \left(2 - \epsilon\right) \left(- \epsilon \left(2 - \epsilon\right) + 1\right) + 0.125 \left(- \epsilon \left(2 - \epsilon\right) + 1\right)^{2}}
    \\ E(FF_{OBB}) &= \frac{0.421875 \epsilon^{4} + 0.3515625 \epsilon^{3} \cdot \left(1 - \epsilon\right) + 0.14453125 \epsilon^{2} \left(1 - \epsilon\right)^{2} + 0.0390625 \epsilon \left(1 - \epsilon\right)^{3} + 0.03125 \left(1 - \epsilon\right)^{4}}{1.6875 \epsilon^{4} + 1.40625 \epsilon^{3} \cdot \left(1 - \epsilon\right) + 0.578125 \epsilon^{2} \left(1 - \epsilon\right)^{2} + 0.15625 \epsilon \left(1 - \epsilon\right)^{3} + 0.03125 \left(1 - \epsilon\right)^{4}}
    \\ E(FF_{SBB}) &= \frac{0.01953125 \epsilon^{4} + 0.05859375 \epsilon^{3} \cdot \left(1 - \epsilon\right) + 0.087890625 \epsilon^{2} \left(1 - \epsilon\right)^{2} + 0.0390625 \epsilon \left(1 - \epsilon\right)^{3} + 0.03125 \left(1 - \epsilon\right)^{4}}{0.0703125 \epsilon^{4} + 0.234375 \epsilon^{3} \cdot \left(1 - \epsilon\right) + 0.2890625 \epsilon^{2} \left(1 - \epsilon\right)^{2} + 0.15625 \epsilon \left(1 - \epsilon\right)^{3} + 0.03125 \left(1 - \epsilon\right)^{4}}
    \\ E(FF_{OBP}) &= \frac{0.0390625 \epsilon^{2} \left(2 - \epsilon\right)^{2} + 0.03125 \epsilon \left(2 - \epsilon\right) \left(- \epsilon \left(2 - \epsilon\right) + 1\right) + 0.03125 \left(- \epsilon \left(2 - \epsilon\right) + 1\right)^{2}}{0.140625 \epsilon^{2} \left(2 - \epsilon\right)^{2} + 0.09375 \epsilon \left(2 - \epsilon\right) \left(- \epsilon \left(2 - \epsilon\right) + 1\right) + 0.03125 \left(- \epsilon \left(2 - \epsilon\right) + 1\right)^{2}}
    \\ E(NS_{OBB}) &= \frac{1.6875 \epsilon^{4} + 1.40625 \epsilon^{3} \cdot \left(1 - \epsilon\right) + 0.578125 \epsilon^{2} \left(1 - \epsilon\right)^{2} + 0.15625 \epsilon \left(1 - \epsilon\right)^{3}}{1.6875 \epsilon^{4} + 1.40625 \epsilon^{3} \cdot \left(1 - \epsilon\right) + 0.578125 \epsilon^{2} \left(1 - \epsilon\right)^{2} + 0.15625 \epsilon \left(1 - \epsilon\right)^{3} + 0.03125 \left(1 - \epsilon\right)^{4}}
    \\ E(NS_{SBB}) &= \frac{0.06640625 \epsilon^{4} + 0.234375 \epsilon^{3} \cdot \left(1 - \epsilon\right) + 0.2578125 \epsilon^{2} \left(1 - \epsilon\right)^{2} + 0.15625 \epsilon \left(1 - \epsilon\right)^{3}}{0.0703125 \epsilon^{4} + 0.234375 \epsilon^{3} \cdot \left(1 - \epsilon\right) + 0.2890625 \epsilon^{2} \left(1 - \epsilon\right)^{2} + 0.15625 \epsilon \left(1 - \epsilon\right)^{3} + 0.03125 \left(1 - \epsilon\right)^{4}}
    \\ E(NS_{OBP}) &= \frac{0.1328125 \epsilon^{2} \left(2 - \epsilon\right)^{2} + 0.078125 \epsilon \left(2 - \epsilon\right) \left(- \epsilon \left(2 - \epsilon\right) + 1\right)}{0.140625 \epsilon^{2} \left(2 - \epsilon\right)^{2} + 0.09375 \epsilon \left(2 - \epsilon\right) \left(- \epsilon \left(2 - \epsilon\right) + 1\right) + 0.03125 \left(- \epsilon \left(2 - \epsilon\right) + 1\right)^{2}}
\end{align}

}

\section{Representation theory}

\subsection{Representation theory background}\label{app:rep theory}

A \emph{representation} of a group $G$ is a pair $(\rho, V)$ of a vector space $V$ and a group homomorphism $\rho: G\rightarrow GL(V)$. 
By abuse of notation, a representation is often referred to by the vector space $V$ alone, especially if the map $\rho$ is understood from context. 
For example, if $G\subseteq GL(d)$ is a matrix group, it has a representation on $\mathbb{C}^d$ by matrix-vector multiplication, formally given by $(\rho, V)$ where $V=\mathbb{C}^d$ and $\rho$ is the identity map on $G$. 
We will often simply refer to $\mathbb{C}^d$ as the \emph{natural representation} of $G$ and omit mention of the identity map $\rho$. 

Two representations $(\rho_1, V_1), (\rho_2, V_2)$ are \emph{equivalent} if there exists a map $\phi: V_1\rightarrow V_2$
that is an isomorphism as a linear map and also commutes with the group action: namely, for all $v_1\in V_1$ and $g\in G$, 
    $\phi(\rho_1(g)v_1) = \rho_2(g) \phi(v_2).$ 
This is the relevant notion of isomorphism for representations, and (when the maps $\rho_1, \rho_2$ are understood) we write $V_1\cong V_2$. 

Given a representation $(\rho, V)$ of $G$, a \emph{subrepresentation} is a vector space $W\subseteq V$ such that, for all $g\in G$ and $w\in W$, $\rho(g)w\in W$. 
This is a new representation of $G$, defined by $(\rho', W)$, where 
    $\rho'(g) = \rho(g)|_W.$  
Similarly, given two groups $G_1\subseteq G_2$ and a representation $(\rho, V)$ of the larger group $G_2$, we obtain the \emph{restricted representation} of the smaller group $G_1$ as $(\rho|_{G_1}, V)$, where 
    $\rho|_{G_1}(g) = \rho(g)$
is simply the restriction of $\rho$ to $G_1$. 

Given two representations $(\rho_1, V_1), (\rho_2, W_2)$ of $G$, we may form the \emph{direct sum representation} $(\rho, V)$ where $V = V_1\oplus V_2$ and, for $g\in G$, $v_1\in V_1$, and $v_2\in V_2$, 
    $\rho(g)(v_1 \oplus v_2) = \rho(g)v_1 \oplus \rho(g) v_2.$ 
We also have the \emph{tensor product representation}, with $V= V_1\otimes V_2$ and 
    $\rho(g)(v_1\otimes v_2) = (\rho_1(g)v_1)\otimes (\rho_2(g)v_2).$ 

A representation $(\rho, V)$ is \emph{irreducible} if it has no nonzero proper subrepresentations: 
in other words, if $W\subseteq V$ is a subrepresentation, then $W=0$ or $W=V$. 
The term "irreducible representation" is often shortened to \emph{irrep}. 
A representation $(\rho, V)$ of $G$ is \emph{completely reducible} if it decomposes into a direct sum of irreps: 
    $V = W_1\oplus W_2\oplus \cdots \oplus W_k,$ 
where each $W_i$ is an irrep for $G$. 
Typically we are only concerned with such decompositions up to equivalence, writing
    $V \cong \bigoplus_i V_i^{\oplus n_i},$ 
where each $V_i$ is irreducible, $V_i\cong V_j$ iff $i=j$, and $n_i$ is called the \emph{multiplicity} of the irrep $V_i$. 
For the groups of interest to us here, specifically finite and compact groups, every finite-dimensional representation $(\rho, V)$ is completely reducible.

\subsection{Representation theory and Proposition~\ref{prop:designs-equivalent-defns}}
\label{app:rep and design exposition}

We now briefly discuss the characterization of $t$-designs in terms of representations, studied in Proposition~\ref{prop:designs-equivalent-defns}. 

Let $\mathcal{V}$ be a $d$-dimensional vector space, and let $U(\mathcal{V})$ be the corresponding unitary group, the group of all unitary transformations of $\mathcal{V}$. 
By picking a basis, we may identify $\mathcal{V}$ with $\mathbb{C}^d$ and $U(\mathcal{V})$ with $U(d)$. 
As discussed above, we have the natural representation $(\tau, \mathcal{V})$ of $U(\mathcal{V})$ on $\mathcal{V}$, where $\tau$ is simply the identity map. 
This representation $\tau$ is in fact irreducible. 
We then take the tensor product of $t$ copies of the natural representation, determining the tensor product representation (or \emph{diagonal representation}) $(\tau^t, \mathcal{V}^{\otimes t})$, given by $\tau^t(U) = U^{\otimes t}$. 
This has a decomposition into irreps for $U(\mathcal{V})$, 
\begin{align}\label{eq:tensor irrep decomp}
    \mathcal{V}^{\otimes t} \cong \bigoplus_i V_{i}^{\oplus n_i}, 
\end{align}
where, as above, each $V_i$ is irreducible and $V_i\cong V_j$ iff $i=j$. 
Proposition~\ref{prop:designs-equivalent-defns} is concerned with the sum of the multiplicities, $\sum_i n_i$. 
For $t=1$, $\mathcal{V}$ is already irreducible, as mentioned above, so $\sum_i n_1 = 1$. 
For $t=2$, $\mathcal{V}^{\otimes 2}$ decomposes into a direct sum of two inequivalent irreps, namely the symmetric and anti-symmetric tensors, so $\sum_i n_i = 2$. 

Now let $G$ be a closed subgroup $G\subseteq U(\mathcal{V})$. 
The representation $\mathcal{V}^{\otimes t}$ restricts to a representation of $G$, which then has its own decomposition into irreps for $G$. 
In fact, it suffices to consider the restriction of each $V_i$ to a representation of $G$. 
The content of Proposition~\ref{prop:designs-equivalent-defns} is that $G$ is a $t$-design if and only if, for each $V_i$ appearing with nonzero multiplicity in \eqref{eq:tensor irrep decomp}, $V_i$ is also irreducible as a representation of $G$. 

\subsection{Proof of Theorem~\ref{thm:no 2 designs}}\label{app:continuous}

We now restate and prove Theorem~\ref{thm:no 2 designs}. 

\begin{theorem}
    Let $\mathcal{V}\cong\mathbb{C}^d$, with $d\geq 2$. 
    Let $G$ be an infinite closed subgroup of $U(d)=U(\mathcal{V})$. If $G$ is a unitary $2$-design, then $SU(d)\subseteq G$. In particular, any set of unitaries generating $G$ is a universal set. 
\end{theorem}
\begin{proof}
    Since $G$ is an infinite closed subgroup of $U(d)$, the closed subgroup theorem implies that it is an infinite compact Lie group, with a nonzero complexified Lie algebra $\mathfrak{g}$. 
    By compactness, the Peter-Weyl theorem implies that $G$'s finite-dimensional representations are completely reducible. 
    Since we assume that $G$ is a $2$-design, it is also a $1$-design; which implies $\mathcal{V}$ is an irrep for $G$ (recall Proposition~\ref{prop:designs-equivalent-defns}). Since $\dim \mathcal{V} = d > 1$, this implies that $G$ is not abelian: equivalently, the $G$-subrepresentation $[\mathfrak{g}, \mathfrak{g}]$ is nonzero. 
    Further, since commutators are traceless, we have $0\neq [\mathfrak{g}, \mathfrak{g}]\subseteq\mathfrak{sl}(d)$. 
    These are the main facts we will need below. 
    
    By Proposition~\ref{prop:designs-equivalent-defns}, since $G$ is a $2$-design,  $\mathcal{V}\otimes \mathcal{V}$ decomposes into a direct sum of exactly two irreps for $G$. 
    Note that by Schur's lemma, this is equivalent to the following condition on the commutant: $\dim_{\mathbb{C}}\left(\End_G(\mathcal{V}\otimes \mathcal{V})\right) = 2.$ 
    We now apply a chain of standard vector space isomorphisms:
    \begin{align}
        \End_G(\mathcal{V}\otimes \mathcal{V}) \cong \left( (\mathcal{V}\otimes \mathcal{V}) \otimes (\mathcal{V}\otimes \mathcal{V})^*\right)^G
        \cong \left( (\mathcal{V}\otimes \mathcal{V}^*) \otimes (\mathcal{V} \otimes \mathcal{V}^*)^* \right)^G
        \cong \End_G(\mathcal{V}\otimes \mathcal{V}^*).
    \end{align}

    In particular, this implies $\dim_{\mathbb{C}}\End_G(\mathcal{V}\otimes \mathcal{V}^*) = 2.$ 
    Now we follow Ref. \citenum{katz2004larsen} to show that this condition gives universality. 
    Applying Schur's lemma in the same way as above, if 
        $\mathcal{V}\otimes \mathcal{V}^*\cong \bigoplus_{i} \mathcal{V}_i^{\oplus m_i}$ 
    is a decomposition of $\mathcal{V}\otimes \mathcal{V}^*$ into distinct irreps for $G$, then 
    \begin{align*}
        2 = \dim_{\mathbb{C}}\End_G(\mathcal{V}\otimes \mathcal{V}^*) = \sum_i m_i^2.
    \end{align*}
    Since the $m_i$ are positive integers, we conclude that $\mathcal{V}\otimes \mathcal{V}^*$ decomposes into exactly two irreducible representations for $G$, each of multiplicity $1$. Further, we can identify them exactly. 
    Note that $\mathcal{V}\otimes \mathcal{V}^*$ is equivalent to the adjoint representation of $U(d)$ on its complexified Lie algebra $\mathfrak{gl}(d)$. 
    As a $U(d)$-representation (and therefore also a $G$-representation), $\mathcal{V}\otimes \mathcal{V}^*$ decomposes as follows:
    \begin{align}
        \mathcal{V}\otimes \mathcal{V}^*\cong \mathfrak{gl}(d) = \mathfrak{sl}(d)\oplus \mathbb{C}I,
    \end{align}
    where $I$ is the $d\times d$ identity matrix. 
    Both $\mathfrak{sl}(d)$ and $\mathbb{C}I$ are nonzero representations of $G$ (since $d\geq 2$), so by the above they must be irreducible $G$-representations. 
    We proved above that the derived subalgebra $[\mathfrak{g}, \mathfrak{g}]$ is a nonzero $G$-subrepresentation of $\mathfrak{sl}(d)$; 
    since $\mathfrak{sl}(d)$ is irreducible, we must have
    $\mathfrak{sl}(d) = [\mathfrak{g}, \mathfrak{g}]\subseteq \mathfrak{g}.$ 
    In particular, taking the skew-Hermitian part, we see that the real Lie algebra of $G$ contains $\mathfrak{su}(d)$, and by applying the matrix exponential, $SU(d)\subseteq G$. 
\end{proof}

\subsection{Proof of Proposition~\ref{prop:1 design}}\label{app:1 design}

\begin{proposition}
    Let $n,m\geq 2$. 
    The group $U_{n,m}$ of linear optical unitaries is a $1$-design in $U(\mathcal{H})$. 
\end{proposition}
\begin{proof}
Recall from \eqref{eq:LOU map} the isomorphism $U(m)\rightarrow U_{n,m}$, 
and let $D = \textnormal{diag}(d_1, \dots, d_m)\in U(m)$. 
For any Fock basis vector $\ket{n_1, \dots, n_m}\in\mathcal{H}$, $D$ acts by 
$D: \ket{n_1, \dots, n_m}\mapsto d_1^{n_1}\cdots d_m^{n_m} \ket{n_1, \dots, n_m}.$ 
In the language of representation theory, this tells us that 
$\ket{n_1, \dots, n_m}$ is a \emph{weight vector} with weight $(n_1, \dots, n_m)$. 
In particular, 
$\mathcal{H}$ is a highest weight representation for $U(m)\cong U_{n,m}$ 
with \emph{highest weight} $(n, 0, \dots, 0)$. 
Therefore, $\mathcal{H}$ contains a copy of $V_{(n, 0, \dots, 0)}$, the irreducible representation of $U(m)$ with highest weight $(n, 0, \dots, 0)$. 
One may compute using the Weyl dimension formula (see e.g. Ref. \citenum{humphreys2012introduction} Corollary 24.3) that 
    $\dim V_{(n, 0, \dots, 0)} = \binom{n+m-1}{n} = d_{n,m} = \dim \mathcal{H},$ 
so the two spaces must be equal. 
\end{proof}

To avoid appealing to general character theory, one may instead show directly that, for any modes $i,j$, the operators $a_i^\dagger a_j$ are in $\mathfrak{u}_{n,m}$, the Lie algebra of the linear optical unitaries. 
These operators move a photon from mode $j$ to mode $i$ and are easily seen to act transitively on the Fock basis vectors. 
This fact, combined with the fact that $\mathcal{H}$ is a weight representation for $U_{n,m}$ with one-dimensional weight spaces spanned by the Fock basis vectors, shows that $\mathcal{H}$ is irreducible. 

\subsection{Discussion of Theorem~\ref{thm:zimboras}}\label{app:zimboras}

We note the following characterization of the vector $\ket{\zeta}$ in Theorem~\ref{thm:zimboras}. This was mentioned in the proof of Lemma 1 in the Supplementary Material to Ref. \citenum{oszmaniec2017universal}. 
One can show that the $U(2)$-module $\mathcal{H}\otimes\mathcal{H}$ has a unique irreducible subrepresentation with highest weight $0$ (necessarily $1$-dimensional), and $\ket{\zeta}$ is a nonzero vector in that subrepresentation. 
This gives a decomposition into subrepresentations, 
\begin{align}\label{eq:tensor decomp}
    \mathcal{H}\otimes\mathcal{H} = W\oplus \mathbb{C}\ket{\zeta}.
\end{align}
(Where $W$ is generally not irreducible.) 
In view of Corollary~\ref{cor:2 design}, we may view Theorem~\ref{thm:zimboras} in the $m=2$ case as saying that, for any operator $V$ that fails to preserve the decomposition \eqref{eq:tensor decomp} into subrepresentations, 
the addition of $V$ to $U_{n,m}$ turns $\mathcal{H}\otimes\mathcal{H}$ into an irreducible representation. 

Perhaps more intuitively, we may rephrase the above as saying that $\ket{\zeta}$ is (up to scalar) the unique element of $\mathcal{H}\otimes\mathcal{H}$ with the property that $\ket{\zeta}$ is an eigenvector of $U\otimes U$ for all $U\in U_{n,m}$. 
(Equivalently, $\ket{\zeta}$ is annihilated by the action of $\mathfrak{su}_{n,m}$, (the traceless subalgebra of) the Lie algebra of $U_{n,m}$.) 
This makes it clear that \eqref{eq:zeta comm} fails for $V\in U_{n,m}$. 
Thus, for the $m=2$ case, one does not even need to check that $V\not\in U_{n,m}$: this is implied by \eqref{eq:zeta comm}. 

\subsection{Proof of Theorem~\ref{thm:snap}}\label{app:snap}

\begin{theorem}
Let $n,m\geq 2$, $0\leq k\leq n$, $0\leq s\leq m-1$. 
Let $\theta\in\mathbb{R}$ with $\theta\not\equiv 0\mod 2\pi$. 
Then the group generated by $U_{n,m}$ and $S_{k,s}(\theta)$ is universal. 
\end{theorem}
\begin{proof}
If $m>2$, Theorem~\ref{thm:zimboras} and the fact that $S_{k,s}(\theta)\not\in U_{n,m}$ implies universality. 
Thus we consider the case $m=2$. Without loss of generality, we consider the SNAP gate on mode $0$, $S_{k,0}(\theta)$. We now directly compute $S_{k,0}(\theta)^{\otimes 2}\ket{\zeta}$ and show that it is not a multiple of $\ket{\zeta}$, implying \eqref{eq:zeta comm}. 
For technical reasons, we begin with the assumption $n\neq 2k$, so that $k\neq n-k$. We have
\begin{align}
    S_{k,0}(\theta)^{\otimes 2}\ket{\zeta} &= \left(\exp(i\theta \ketbra{k}_0)\otimes \exp(i\theta \ketbra{k}_0) \right)\sum_{a=0}^n (-1)^a \ket{a, n-a}\otimes\ket{n-a, a}
    \\&= \sum_{a\neq k,n-k} (-1)^a \ket{a, n-a}\otimes\ket{n-a, a} + e^{i\theta}\sum_{a= k,n-k} (-1)^a \ket{a, n-a}\otimes\ket{n-a, a}.
\end{align}
Then 
\begin{align}
    S_{k,0}(\theta)^{\otimes 2}\ket{\zeta} - \ket{\zeta} = (e^{i\theta}-1)\sum_{a= k,n-k} (-1)^a \ket{a, n-a}\otimes\ket{n-a, a},
\end{align}
which is not a multiple of $\ket{\zeta}$ as long as $\theta\not\equiv 0\mod 2\pi$. 
If $n=2k$, we similarly obtain 
\begin{align}
    S_{k,0}(\theta)^{\otimes 2}\ket{\zeta} - \ket{\zeta} = (e^{2i\theta}-1) (-1)^k \ket{k, k}\otimes\ket{k,k}. 
\end{align}
Theorem~\ref{thm:zimboras} then gives the result. 
\end{proof}

\end{document}